\UseRawInputEncoding
\documentclass[lettersize,journal]{IEEEtran}
\usepackage{amsmath, amsthm}
\usepackage[OT1]{fontenc}
\usepackage[english]{babel}
\usepackage[nolist]{acronym} 
\usepackage{mfirstuc}
\usepackage[utf8]{inputenc}

\usepackage{algorithm}
\usepackage{algorithmicx}
\usepackage{algpseudocode}

\newcommand{\AlgRule}{%
    \Statex \rule{0.95\linewidth}{0.1pt}%
}
\algrenewcommand\algorithmicprocedure{\textbf{subroutine}}
\usepackage{nicefrac}

\usepackage{array}
\usepackage{caption}
\usepackage{textcomp}
\usepackage{stfloats}
\usepackage{url}
\usepackage{verbatim}
\usepackage{graphicx}
\usepackage{xcolor}

\usepackage{tikz}
\usepackage{pgfplots, pgfplotstable}
\usepackage{amssymb}
\usepackage{subfigure}
\usepackage{cite}
\usepackage{dsfont}

\usepackage{makecell}
\usepackage{booktabs}

\usepackage{arydshln}
\newcolumntype{C}[1]{>{\centering\let\newline\\\arraybackslash\hspace{0pt}}m{#1}}
\newcolumntype{L}{>{\arraybackslash}p{0.7\linewidth}}
\newcolumntype{N}{@{}m{0pt}@{}}

\usepgfplotslibrary{statistics}
\usepgfplotslibrary{groupplots}

\usepackage{multirow}
\usepackage{standalone}
\usepackage{balance}
\usepackage[capitalise]{cleveref}

\Crefformat{figure}{#2Fig.~#1#3}
\Crefmultiformat{figure}{Fig.~#2#1#3}{ and~#2#1#3}{, #2#1#3}{ and~#2#1#3}

\newtheorem{lemma}{Lemma}
\newtheorem{theorem}{Theorem}

\pgfplotsset{
    legend image with text/.style={
        legend image code/.code={%
            \node[anchor=center] at (0.3cm,0cm) {#1};
        }
    },
    label style = {font=\fontsize{9pt}{7.2}\selectfont},
    tick label style = {font=\fontsize{7pt}{7.2}\selectfont}
}

\makeatletter
\newcommand{\AcCaps}[1]{%
  \expandafter\AC@get
  \csname fn@#1\endcsname
  \AcCapsAux{#1}%
}
\newcommand{\AcCapsAux}[2]{\capitalisewords{#2}}
\makeatother

\begin{acronym}
\acro{ML}{maximum likelihood}
\acro{BER}{bit error rate}
\acro{FER}{frame error rate}
\acro{SNR}{signal-to-noise-ratio}
\acro{BPSK}{binary phase shift keying}
\acro{AWGN}{additive white Gaussian noise}
\acro{NextG}{next generation}
\acro{5G}{fifth generation}
\acro{SC}{successive cancellation}
\acro{SCL}{successive cancellation list}
\acro{BP}{belief propagation}
\acro{sp-node}{special-node}
\acro{CRC}{cyclic redundancy check}
\acro{Spec-SC}{speculative successive cancellation}
\acro{LLR}{log-likelihood ratio}
\acro{NR}{New Radio}
\acro{SPC}{single-parity-check}
\acro{CPU}{central processing unit}
\acro{GPU}{graphics processing unit}
\acro{FPGA}{field-programmable gate array}
\acro{ASIC}{application-specific integrated circuit}
\acro{PL}{processing lane}
\end{acronym}

\renewcommand{\vec}[1]{\boldsymbol{#1}}
\newcommand{\vecs}[1]{\boldsymbol{#1}}

\newcommand{\dv}{\vec{d}}

\newcommand{\mv}{\vec{m}}
\newcommand{\nv}{\vec{n}}

\newcommand{\pv}{\vec{p}}
\newcommand{\qv}{\vec{q}}

\newcommand{\sv}{\vec{s}}

\newcommand{\uv}{\vec{u}}

\newcommand{\xv}{\vec{x}}
\newcommand{\yv}{\vec{y}}

\newcommand{\ellv}{\vecs{\ell}}

\newcommand{\Fm}{\vec{F}}
\newcommand{\Gm}{\vec{G}}

\newcommand{\Mm}{\vec{M}}

\newcommand{\Fc}{{\cal F}}

\newcommand{\Ic}{{\cal I}}

\newcommand{\Pc}{{\cal P}}

\newcommand{\rtse}{\mathrm{r}}
\newcommand{\ltse}{\mathrm{l}}
\newcommand{\ptse}{\mathrm{p}}

\newcommand{\LB}{\left(}
\newcommand{\RB}{\right)}

\newcommand{\sgn}{\mathop{\mathrm{sgn}}}

\newcommand{\EbNo}{E_\mathrm{b}/N_0}
\newcommand{\XOR}{\oplus}

\definecolor{matplottikz-color1}{HTML}{144b32}
\definecolor{matplottikz-color2}{HTML}{db7622}
\definecolor{matplottikz-color3}{HTML}{314c82}
\definecolor{matplottikz-color4}{HTML}{c71a1f}

\definecolor{matplottikz-color5}{HTML}{8a63a9}
\definecolor{matplottikz-color6}{HTML}{2a9d9a}
\definecolor{matplottikz-color7}{HTML}{c39b26}
\definecolor{matplottikz-color8}{HTML}{b05c85}
\definecolor{matplottikz-color9}{HTML}{657a43}
\definecolor{matplottikz-color10}{HTML}{536b78}
\definecolor{matplottikz-color11}{HTML}{996247}
\definecolor{matplottikz-color12}{HTML}{5a99b5}

\definecolor{matblue}{HTML}{314C82}
\definecolor{matorange}{HTML}{DB7622}
\definecolor{matgreen}{HTML}{318D43}
\definecolor{matred}{HTML}{C71A1F}

\usetikzlibrary{plotmarks}

\newcommand{\plotmarker}[1]{%
  \ifcase#1\relax
  \or *%
  \or square*%
  \or triangle*%
  \or diamond*%
  \or pentagon*%
  \or otimes*%
  \or triangle%
  \or square%
  \or diamond%
  \or pentagon%
  \or asterisk%
  \or 10-pointed star%
  \else *%
  \fi
}

\pgfplotsset{
    method/SC/.style={
        color=matplottikz-color1,
        mark=\plotmarker{1},
        mark size=1.6,
        mark options={solid, fill=white},
        line width=1.2pt
    },
    method/SpecSC/.style={
        color=matplottikz-color2,
        mark=\plotmarker{2},
        mark size=1.6,
        mark options={solid, fill=white},
        line width=1.2pt
    },
    method/BP20/.style={
        color=matplottikz-color3,
        mark=\plotmarker{3},
        mark size=1.6,
        mark options={solid, fill=white},
        line width=1.2pt
    },
    method/N1024/.style={
        color=matplottikz-color1,
        mark=pentagon*,
        mark size=1.6,
        mark options={solid, fill=white},
        line width=1.2pt
    },
    method/N512/.style={
        color=matplottikz-color2,
        mark=\plotmarker{2},
        mark size=1.6,
        mark options={solid, fill=white},
        line width=1.2pt
    },
    method/N256/.style={
        color=matplottikz-color3,
        mark=\plotmarker{3},
        mark size=1.6,
        mark options={solid, fill=white},
        line width=1.2pt
    },
    method/N128/.style={
        color=matplottikz-color4,
        mark=\plotmarker{4},
        mark size=1.6,
        mark options={solid, fill=white},
        line width=1.2pt
    },
    method/BP/.style={
        color=matplottikz-color5,
        mark=x,
        mark size=1.6,
        mark options={solid, fill=white},
        line width=1.2pt
    }
}

\pgfplotsset{
    rate/quarter/.style={
        dotted
    },
    rate/half/.style={
        dashed
    },
    rate/seveneighths/.style={
        solid
    }
}

\pgfplotsset{
    graph/SpeedUp/.style={
        scale=0.8,
        width=0.9\columnwidth,
        height=0.75\columnwidth,
        xlabel={$\EbNo$ (dB)}, 
        ylabel={Average Speedup per Frame}, 
        ylabel style={yshift=-15pt},
        grid=both,
        grid style=dashed,
        xmajorgrids=true,
        ymajorgrids=true,
        minor x tick num=1,
        legend style={
            font=\fontsize{7pt}{7.2}\selectfont,
            cells={anchor=west},
            column sep=2mm,
            legend columns=3,
            at={(0.5,1.15)},
            anchor=south
        },
    },
    graph/SpeedUp/.belongs to family=/pgfplots/scale,
    graph/SuccessRate/.style={
        scale = 0.75, 
        width=0.9\columnwidth,
        height=0.75\columnwidth,
        xlabel={Node ID},
        ylabel={Success Rate},
        ylabel style={yshift=-10pt},
        grid=both,
        grid style=dashed,
        xmajorgrids=true,
        ymajorgrids=true,
        minor x tick num=1,
        ymin=0,
        ymax=1.0,
        xmin=0,
        xmax=820,
        scaled x ticks=false,
    },
    graph/SuccessRate/.belongs to family=/pgfplots/scale,
    graph/BERFER/.style={
        scale = 1,
        ymode=log,
        minor x tick num=9,
        xmin=0.85,
        xmax=6.1,
        xlabel={$E_\mathrm{b}/N_0$ (dB)}, 
        ylabel style={yshift=-5pt},
        ytick={1e-0, 1e-1, 1e-2, 1e-3, 1e-4, 1e-5, 1e-6, 1e-7, 1e-8},
        grid=both,
        ymajorgrids=true,
        xmajorgrids=true,
        grid style=dashed,
        width=0.9\columnwidth,
        height=0.75\columnwidth,
        legend style={
            anchor=center,
            cells={anchor=west},
            column sep= 2mm,
            thick,
            font=\fontsize{7pt}{7.2}\selectfont,
        },
    },
    graph/BERFER/.belongs to family=/pgfplots/scale,
}

\begin{document}
\bstctlcite{IEEEcontrol}
\thispagestyle{empty}

\title{Speculative Successive Cancellation Decoding \\ of Polar Codes}

\author{
Ryan Seah, Marvin R\"ubenacke, and Warren J. Gross \\
\thanks{Manuscript received (TBC); revised (TBC); accepted (TBC). Date of current version (11th November 2026). (Corresponding author: Ryan Seah.)}
\thanks{Ryan Seah, Marvin R\"ubenacke, and Warren J. Gross are with the Department of Electrical and Computer Engineering, McGill University, Montr{\'e}al, Qu{\'e}bec, Canada. (e-mail: ryan.seah@mail.mcgill.ca; \{marvin.ruebenacke, warren.gross\}@mcgill.ca).}
}

\markboth{IEEE Journal on ...}%
{Seah \MakeLowercase{\textit{et al.}}: title}

\maketitle

\begin{abstract}
Polar codes achieve channel capacity as block length increases, but this asymptotic advantage comes at the cost of decoding speed: the conventional \ac{SC} algorithm is inherently sequential, which limits its practicality for high-throughput applications. This work addresses that limitation by introducing the \emph{\ac{Spec-SC}} framework, which incorporates speculative execution into the \ac{SC} decoding process. While \ac{sp-node}-based decoders enable partial parallelism for specific, predefined structures, \ac{Spec-SC} extends this capability by permitting parallel execution in general nodes that do not match any \acs{sp-node} pattern. A built-in verification mechanism ensures that this speculation does not compromise decoding correctness. Simulation results show that \ac{Spec-SC} achieves the same \acs{BER} and \acs{FER} performance as conventional \ac{SC} decoding, while reducing the average number of sequentially traversed nodes by up to $69\%$ for a polar code of length $N=4096$. A cycle-level analysis further shows that, for $\Pc(4096,3584)$ at $5\,\mathrm{dB}$, a heterogeneous four-core configuration achieves a $1.57\times$ decoding speedup using a total of $2944$ processing lanes, representing a $43.75\%$ increase over the $N_{\mathrm{lane}}=2048$ one-core \ac{SC} baseline. The speedup increases to $2.5\times$ at $7\,\mathrm{dB}$, demonstrating the increasing effectiveness of speculation at higher $\EbNo$. These results demonstrate that \ac{Spec-SC} enables additional algorithm-level parallelism beyond conventional \ac{sp-node}-based decoding, motivating its further investigation for high-throughput and high-reliability communication systems.
\end{abstract}


\acresetall
\section{Introduction}

\IEEEPARstart{P}{olar} codes are the first class of error-correcting codes that are provably capacity-achieving for symmetric binary-input memoryless channels under \ac{SC} decoding \cite{arikan2009channel}. 
Their recursive structure enables efficient hardware implementations and facilitates systematic extensions to more powerful decoding algorithms, such as \ac{SCL} decoding~\cite{sarkis2014fast, sarkis2015fast, tal2015list, fan2015low, hashemi2017fast, alamdar2011simplified, balatsoukas2015llr, cocskun2022information} and \ac{CRC}-aided decoding~\cite{ chen2012list, Niu2012CRC, li2012adaptive, sarkis2012data,zhang2015split}.

\begin{figure}[t]
    \centering
    \includegraphics[width=\linewidth]{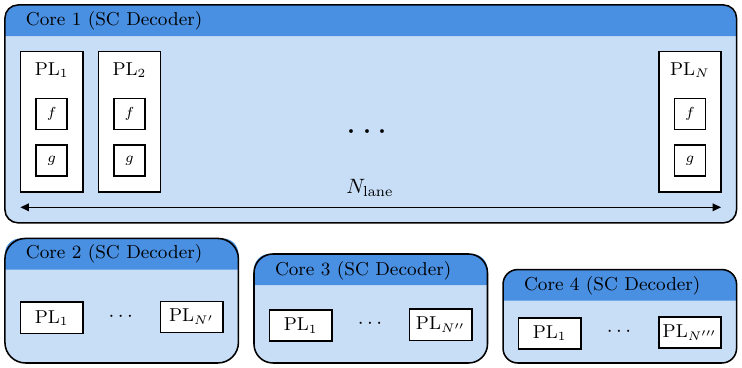}
    \caption{A \textit{multi-core} \ac{SC} decoder architecture, where multiple \ac{SC} cores can operate on the same codeword. Each \ac{SC} core contains $N_{\mathrm{lane}}$ parallel \acp{PL} for $f$ and $g$ computations, with potentially different $N_{\mathrm{lane}}$ across cores. The proposed \acs{Spec-SC} algorithm exploits this architecture by allowing multiple cores to decode dependent branches concurrently through speculative execution, with verification ensuring identical decoding results to conventional \ac{SC}.}
    \label{fig:multi_core}
\end{figure}

While early deployments of polar codes have focused on short block lengths suited for control channels, potential next-generation data channel applications of polar codes increasingly demand longer block lengths, typically $N \geq 4096$~\cite{zhang2023channel, geiselhart20236g, proceedings2024trends, yang2026improved}. Recent studies highlight that long polar codes are promising for high-throughput data transmission, where reliability and spectral efficiency are paramount~\cite{wang2023perturbation, li2025enhanced, yang2026improved}. Beyond wireless communications, optical communication systems and high-speed interconnects also impose stringent requirements on decoding latency, throughput and error-correction performance~\cite{mehmood2021experimental,koike2018irregular}. These applications motivate decoding techniques that achieve low latency while maintaining high reliability. As the code length increases, the polarization effect strengthens, yielding a clearer separation between reliable and unreliable subchannels, and thus offering improved coding gains and near-capacity performance under \ac{SC} decoding~\cite{krieg2025longpolarvsldpc}.

However, decoding long polar codes presents significant implementation challenges. The conventional \ac{SC} decoder follows a strictly sequential decision process, where each bit estimate depends on all previously decoded bits, resulting in inherently limited parallelism, increased decoding latency, and limited throughput as the code length $N$ grows. One effective approach to improving throughput and reducing latency is the exploitation of \acp{sp-node}~\cite{alamdar2011simplified, sarkis2014fast, sarkis2015fast, tong2023fast, ren2022sequence, zheng2021threshold}. These nodes correspond to specific structural patterns, such as rate-0, rate-1, repetition, \ac{SPC}, and sequence-repetition subtrees, which can be decoded in a single operation, thereby collapsing multiple stages of the decoding tree into composite nodes and reducing traversal overhead. Since \ac{SC} decoding latency is primarily determined by the number of sequentially traversed nodes, reducing these traversals improves latency and throughput.

While \ac{sp-node}-based decoding offers substantial improvements in throughput and computational efficiency, it does not fundamentally resolve the sequential dependency inherent to the \ac{SC} algorithm. The presence and distribution of \acp{sp-node} depend strongly on the code construction and are often irregular, leaving portions of the decoding tree where no speedup can be achieved through \ac{sp-node}-based decoding. Moreover, dependencies between \acp{sp-node} persist, preventing full parallel execution and limiting hardware utilization. Consequently, the overall throughput remains dominated by serial segments. 

\Ac{BP} decoding provides another approach to reducing polar-decoding latency by exploiting the parallel structure of the polar factor graph~\cite{arikan2008performance,niu2014polar}. Unlike SC decoding, BP iteratively exchanges soft messages across the graph, allowing many processing elements to operate concurrently and enabling high-throughput implementations \cite{arikan2008performance,yuan2014early,abbas2016high}. However, its error-correction performance depends on iterative convergence and may require many iterations to achieve similar performance to \ac{SC}~\cite{sun2016error}. Moreover, BP decoding of polar codes can exhibit an error floor depending on the code construction and \ac{LLR} clipping~\cite{elkelesh2017errorfloor}.

Other works have sought to reduce \ac{SC} decoding latency through alternative decoding schedules~\cite{bossert2025hiddencodewords,kamenev2023highratefirst}, tree-level parallelism~\cite{kam2021ultralow}, and syndrome-based pruning~\cite{yoo2016efficient}. In particular, the tree-level parallel architecture in~\cite{kam2021ultralow} partitions the decoding tree into parallel subtrees, requiring synchronized pruning structures that constrain the admissible \ac{sp-node} patterns. Syndrome-based pruning~\cite{yoo2016efficient}, on the other hand, bypasses constituent-node decoding when the hard decisions satisfy the corresponding parity constraints. Although these approaches reduce decoding latency, they do not exploit speculative execution to overlap the processing of dependent branches while retaining the conventional \ac{SC} tree structure. 

This work proposes \acf{Spec-SC} decoding as a novel implementation-independent decoding paradigm that enables parallelism with multiple \ac{SC} decoding ``cores'' operating on the same codeword, as shown in~\Cref{fig:multi_core}. The \Ac{Spec-SC} algorithm uses these additional cores to opportunistically decode the right branch before the left branch result is available. Once the left branch completes, a verification test determines whether the speculative result can be accepted. If it passes, the decoder returns the speculated result; otherwise, the right branch is decoded normally. \Ac{Spec-SC} therefore overlaps computations that are sequential in conventional \ac{SC}. Accepted speculation can also make general nodes behave as temporary \acp{sp-node}, adding parallelism to \ac{SC} decoding beyond conventional \ac{sp-node} structures.

This paper makes the following key contributions:
\begin{itemize}
    \item The proposal of the \ac{Spec-SC} algorithm, which uses multiple \ac{SC} cores to speculatively execute dependent branches in parallel, enabling additional parallelism within the \ac{SC} decoder.

    \item A verification mechanism that accepts a speculative result only when it is guaranteed to match conventional \ac{SC} decoding. Thus, \ac{Spec-SC} does not approximate or heuristically modify the \ac{SC} decision; it produces the same output as \ac{SC} for any input \ac{LLR} vector.

    \item A performance analysis demonstrating substantial reductions in sequential decoding-tree traversal while preserving identical error-correction performance.

    \item A cycle-level multi-core analysis demonstrating that the additional parallelism introduced by \ac{Spec-SC} translates into lower decoding latency, achieving up to $2.5\times$ speedup over one-core \ac{SC} under the evaluated configuration, with a $43.75\%$ increase in processing elements.
\end{itemize}

While speculative techniques have been explored in~\cite{mishra2012successive}, these approaches differ fundamentally from the framework proposed in this work. In particular, what is referred to as “speculation” in~\cite{mishra2012successive} is more accurately characterized as look-ahead decoding, wherein multiple candidate inputs to the right branch of the decoding tree are evaluated in parallel. This process effectively corresponds to a localized brute-force enumeration of possible partial sums. For a subtree of size $N' = 2^m$, this results in $2^m$ hypotheses, and equivalence to conventional \ac{SC} decoding is only guaranteed if all such possibilities are considered. Consequently, these methods are typically restricted to small nodes and do not scale efficiently. In contrast, the proposed \ac{Spec-SC} evaluates a single tentative right-branch output and verifies it using a sign-consistency condition. 

The following notation is used throughout the paper:
\begin{itemize}
\item Bold uppercase letters (e.g., $\Mm$) denote matrices.
\item Bold lowercase letters (e.g., $\mv, \ellv$) denote vectors. 
\item $\text{sgn}(\cdot)$: sign function.
\item $|\cdot|$: absolute value.
\item $||$: vector concatenation operator.
\item $\XOR$: bitwise XOR operation.
\end{itemize}

The remainder of this paper is organized as follows. \Cref{sec:preliminaries} provides a brief overview of polar codes, \ac{SC} decoding, and \ac{sp-node}-based decoding techniques. \Cref{sec:spec_sc} introduces the \ac{Spec-SC} decoding framework and its verification mechanism, as well as a theoretical proof of its correctness. \Cref{sec:results} presents simulation results evaluating the performance of the \ac{Spec-SC} decoder. Finally, \Cref{sec:conclusion} concludes the paper.

\section{Preliminaries}
\label{sec:preliminaries}
The following section provides a brief overview of polar codes, \ac{SC} decoding, and \ac{sp-node}-based decoding techniques.

\subsection{Polar Codes}
\label{subsec:polar_code}
A polar code~\cite{arikan2009channel} of length $N = 2^n$ and $K$ information bits is denoted by $\Pc(N,K)$. Its construction partitions the input vector $\uv \in \{0,1\}^N$ based on channel polarization. The $K$ indices corresponding to the most reliable bit-channels form the information set $\Ic$, where the information bits $\uv_{\Ic}$ are placed. The remaining $N-K$ indices constitute the frozen set $\Fc$, with their corresponding bits fixed to a predetermined value, usually zero. 

The encoded codeword $\xv$ is generated through the linear transformation 
\begin{align}
\xv = \uv \Gm_N,
\end{align}
where the generator matrix $\Gm_N$ is the $n$-fold Kronecker power of the kernel $\Fm = \left[\begin{smallmatrix} 1 & 0 \\ 1 & 1\end{smallmatrix}\right]$, i.e., $\Gm_N = \Fm^{\otimes n}$.

In this work, we consider \ac{BPSK} modulation over an \ac{AWGN} channel.
The codeword is mapped to the transmitted vector $\sv = 1 - 2\xv$, and the received vector is
\begin{align}
\yv = \sv + \nv, \qquad \nv \sim \mathcal{N}(\boldsymbol{0}, \sigma^2 \boldsymbol{I}),
\end{align}
where $\nv$ is a noise vector with components drawn from a zero-mean Gaussian distribution with variance $\sigma^2$. The code rate is given by the ratio $R = \nicefrac{K}{N}$.

\subsection{\AcCaps{SC} Decoding}
\label{subsec:scdecoding}

The \ac{SC} decoder forms the foundation of many polar decoding algorithms \cite{arikan2009channel}. It operates recursively over a decoding tree using \acp{LLR} to estimate the transmitted bit sequence. The \ac{SC} algorithm proceeds in a sequential manner, where each node's decision depends on decisions made at preceding nodes, resulting in an inherently serial decoding process.

Let ${\ell}_{i,t}$ denote the \ac{LLR} associated with node $i$ at stage $t$, where the stage index $t$ ranges from $n=\log_2 N$ down to $0$, and the node index $i$ ranges from $0$ to $N - 1$. The process begins at the top of the decoding tree ($t = n$) with the \acp{LLR} ${\ell}_{i,n}$ computed from the received channel values, and proceeds recursively down to the leaf stage ($t = 0$), where final bit decisions are made.

At each stage, two update rules are applied to propagate information through the decoding tree. The first update computes the \ac{LLR} for the left branch using the function $f(\cdot)$, which combines the \acp{LLR} from the next higher stage:
\begin{equation}
\label{eq:f-minsum}
{\ell}_{i,t} = f\!\LB{\ell}_{i,t+1},\, {\ell}_{i+2^t,t+1}\RB,
\end{equation}
where, under the min-sum approximation, $f(\cdot)$ is defined as
\begin{equation}
\label{eq:f}
f({a},{b}) = \sgn({a})\,\sgn({b})\,\min\!\LB|{a}|,\,|{b}|\RB.
\end{equation}

Once the decision ${u}_{i,t}$ for the left branch is made, it is used to compute the \ac{LLR} for the right branch through the function $g(\cdot)$:
\begin{equation}
\label{eq:g}
{\ell}_{i+2^t,t} = g\!\LB\, {\ell}_{i,t+1},\, {\ell}_{i+2^t,t+1}, {u}_{i,t}\RB,
\end{equation}
where
\begin{equation}
\label{eq:gdef}
g({a}, {b}, {c}) = (1 - 2{c}) \cdot {a} + {b}.
\end{equation}
The functions $f(\cdot)$ and $g(\cdot)$ are defined for scalar inputs but can also be applied element-wise when operating on vectors. 

When the traversal reaches the leaf stage ($t = 0$), the decoder makes a hard decision for each bit ${u}_{i,0}$. If the bit index $i$ belongs to the frozen set $\Fc$, its value is predetermined as ${u}_{i,0} = 0$. If it belongs to the information set $\Ic$, the decision is based on the sign of the corresponding \ac{LLR}:
\begin{equation}
\hat{u}_{i,0} =
\begin{cases}
0, & \text{if }  {\ell}_{i,0} \ge 0, \\
1, & \text{otherwise}.
\end{cases}
\end{equation}
At the parent node, the estimates are combined as
\begin{equation}
    \hat{\uv}_\ptse = (\hat{\uv}_\ltse \XOR \hat{\uv}_\rtse) \,||\, \hat{\uv}_\rtse.
\end{equation}
This recursive process continues until all $N$ bits are decoded in the order $\hat{\uv}_{\ltse} \rightarrow \hat{\uv}_{\rtse} \rightarrow \hat{\uv}_{\ptse}$, as shown in \Cref{fig:spec_sc_algo} (left).

\subsection{\AcCaps{sp-node}-Based Decoding}
\label{subsec:node_based_decoding}
To increase the throughput associated with traversing the full decoding tree, the structure can be pruned by identifying \acp{sp-node} based on their frozen bit patterns and collapsing their subtrees. A vector $\dv_{i,t}$, known as the rate profile, is used to identify the nature of the leaf bits descending from a node at stage $t$ and index $i$. In this vector, $0$ signifies a frozen bit, while $1$ indicates an information bit.

Following the simplified decoder designs in \cite{alamdar2011simplified, sarkis2014fast, tong2023fast}, this work implements the following eight \acp{sp-node}:
\begin{enumerate}
    \item \makebox[1.5cm][l]{\textbf{R0}}  $\dv_i = (0, 0, \dots, 0, 0)$
    \item \makebox[1.5cm][l]{\textbf{R1}}  $\dv_i = (1, 1, \dots, 1, 1)$
    \item \makebox[1.5cm][l]{\textbf{REP}} $\dv_i = (0, 0, \dots, 0, 1)$
    \item \makebox[1.5cm][l]{\textbf{SPC}} $\dv_i = (0, 1, \dots, 1, 1)$
    \item \makebox[1.5cm][l]{\textbf{REP-2}} $\dv_i = (0, 0, \dots, 0, 0, 1, 1)$
    \item \makebox[1.5cm][l]{\textbf{SPC-2}} $\dv_i = (0, 0, 1, 1,\dots, 1, 1)$
    \item \makebox[1.5cm][l]{\textbf{RPC}} $\dv_i = (0, 0, 0, 1,\dots, 1, 1)$
    \item \makebox[1.5cm][l]{\textbf{PCR}} $\dv_i = (0, 0, \dots, 0, 1, 1, 1)$
\end{enumerate}
The specialized nodes correspond to component codes with low-complexity maximum-likelihood decoders operating directly on the \ac{LLR} vector. This enables the decoder to skip the recursive traversal for entire subtrees, thereby accelerating the overall decoding process. We refer to this \ac{sp-node}-based decoding as $\operatorname{FastSimplifiedSC}{(\cdot)}$.

\begin{figure}[t]
    \centering
    \includegraphics[width=\linewidth]{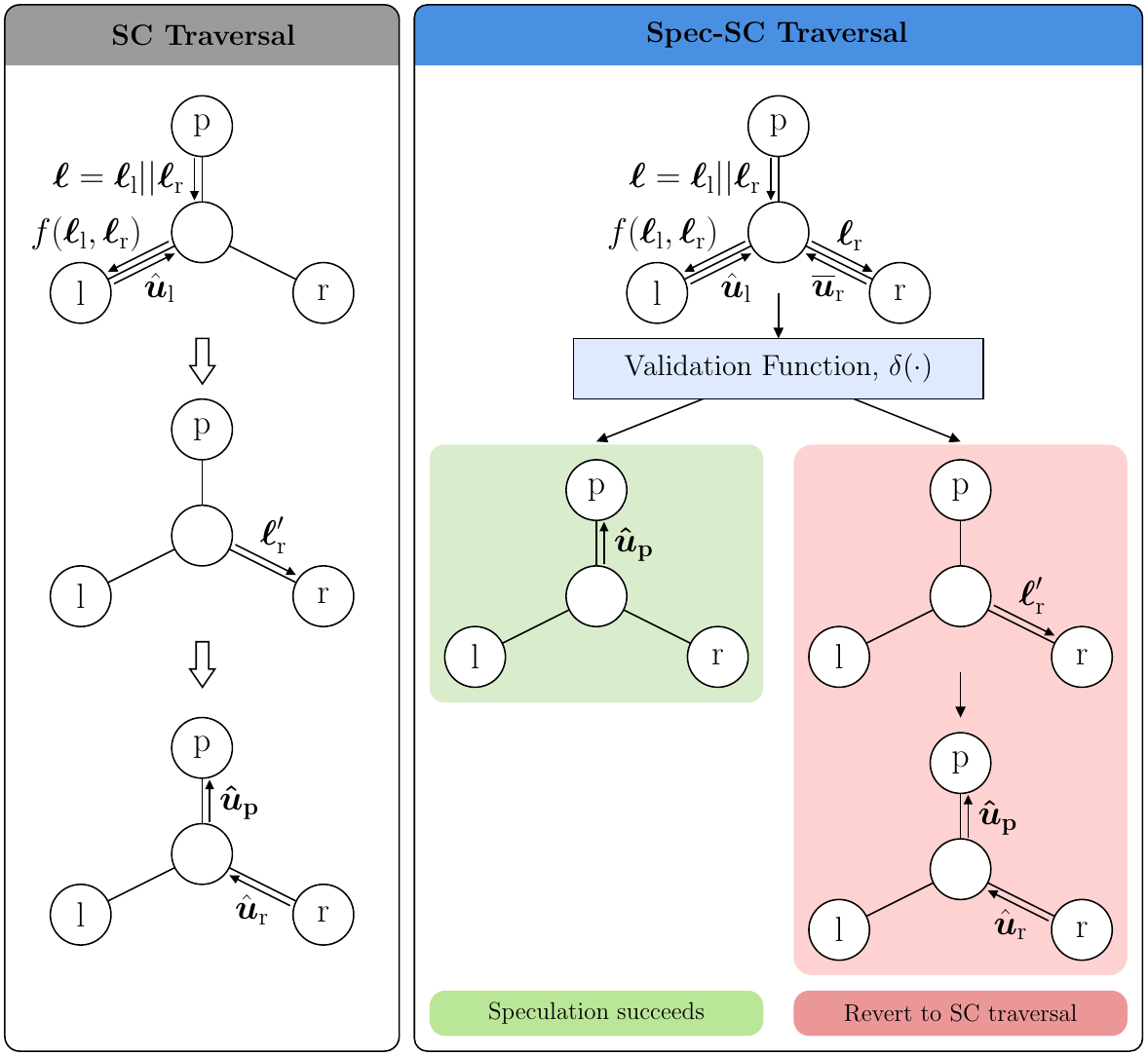}
    \caption{Comparison between standard \acf{SC} traversal and \acf{Spec-SC} traversal. In \ac{SC} decoding (left), node decisions are made sequentially in the order $\hat{\uv}_{\ltse} \rightarrow\hat{\uv}_{\rtse} \rightarrow \hat{\uv}_{\ptse}$. In contrast, \ac{Spec-SC} (right) simultaneously traverses $\hat{\uv}_{\ltse}$ and $\hat{\uv}_{\rtse}$, where $\hat{\uv}_{\rtse} = \bar{\uv}_{\rtse}$ is being speculated. The function $\delta$ validates whether the speculation is successful. If successful, $\hat{\uv}_{\ptse}$ is computed directly; otherwise the process reverts to the conventional \ac{SC} traversal path.}
    \label{fig:spec_sc_algo}
\end{figure}

\section{\AcCaps{Spec-SC} Decoding}
\label{sec:spec_sc}

The \ac{Spec-SC} decoding framework, illustrated in~\Cref{fig:spec_sc_algo} (right), addresses the sequential dependency of conventional \ac{SC} decoding by enabling partial parallel traversal within the decoding tree. Rather than waiting for the left child node $\uv_{\ltse}$ to be fully decoded before proceeding, \ac{Spec-SC} performs simultaneous traversal of both subnodes, $\uv_{\ltse}$ and $\uv_{\rtse}$. The input \ac{LLR} vector at the node is represented as
\begin{align}
\ellv = \ellv_{\ltse} \, || \, \ellv_{\rtse}
\end{align}
where $\ellv_{\rtse}$ is used to generate a speculative decision $\bar{\uv}_{\rtse} = \operatorname{FastSimplifiedSC}(\ellv_{\rtse}, \dv_{\rtse})$. 
After both $\hat{\uv}_{\ltse}$ and $\bar{\uv}_{\rtse}$ are available, the \ac{LLR} for the right branch is updated as
\begin{equation}
\ellv'_{\rtse} = g(\ellv_{\ltse}, \ellv_{\rtse}, \hat{\uv}_{\ltse}).
\end{equation}
To validate the correctness of the speculation, two sign vectors are defined:
\begin{align}
\qv &= \sgn(\ellv'_{\rtse}), \\
\pv &= \sgn(1 - 2\bar{\uv}_{\rtse}).
\end{align}
The verification function $\delta(\cdot)$ is then evaluated as
\begin{equation}
\delta(\qv, \pv) =
\begin{cases}
1, & \text{if } \qv = \pv, \\
0, & \text{otherwise}.
\end{cases}
\end{equation}
If $\delta(\cdot) = 1$, the speculation is deemed successful, and the parent node decision is computed directly as
\begin{equation}
\hat{\uv}_{\ptse} = \LB\hat{\uv}_{\ltse} \XOR \bar{\uv}_{\rtse} \RB \, || \, \bar{\uv}_{\rtse}.
\end{equation}
If $\delta(\cdot) = 0$, the speculative candidate is not accepted, and the right subtree is re-decoded using the corrected \ac{LLR} vector $\ellv'_{\rtse}$. Importantly, wherever speculation is applied, correctness is preserved because a speculative result is accepted only if it satisfies the deterministic verification condition.

\Cref{lemma} shows that this sign test is sufficient to certify the speculative right-branch decision.

\begin{lemma}
\label{lemma}
Let $\ellv=\ellv_{\ltse}\,||\,\ellv_{\rtse}$, and let $\bar{\uv}_{\rtse}$ be a speculative right-subtree candidate. After decoding the left subtree, the right-branch \ac{LLR} is
$\ellv'_{\rtse} = g\!\LB \ellv_{\ltse}, \ellv_{\rtse}, \hat{\uv}_{\ltse} \RB$.

If
\begin{equation}
    \sgn\!\LB1-2\bar{\uv}_{\rtse}\RB
    =
    \sgn\!\LB\ellv'_{\rtse}\RB,
\end{equation}
then \ac{SC} decoding of the right subtree with input \ac{LLR}
$\ellv'_{\rtse}$ returns
\begin{equation}
    \hat{\uv}_{\rtse} = \bar{\uv}_{\rtse}.
\end{equation}
\end{lemma}

\begin{proof}
We prove the result by induction on $\operatorname{length}(\ellv'_{\rtse})$.

For the base case 
$\operatorname{length}(\ellv'_{\rtse})=1$,
the result follows directly from the \ac{SC} leaf decision rule. For an information bit, the sign is equivalent to the hard decision as \ac{SC}. For a frozen bit,
$\bar{\uv}_{\rtse}=0$ by construction, and
the sign condition is consistent with the \ac{SC} frozen bit
decision.

For the induction step, write
$ \ellv'_{\rtse} = \ellv_{\mathrm{rl}} \,||\, \ellv_{\mathrm{rr}}$,
and
$ \bar{\uv}_{\rtse} = \LB \bar{\uv}_{\mathrm{rl}} \XOR \bar{\uv}_{\mathrm{rr}} \RB \,||\, \bar{\uv}_{\mathrm{rr}}$.
By the hypothesis of \Cref{lemma}, component-wise,
\begin{align}
    \sgn\!\LB \ellv_{\mathrm{rl}}[i] \RB &= 1-2\LB\bar{\uv}_{\mathrm{rl}}[i]\XOR\bar{\uv}_{\mathrm{rr}}[i]\RB,\\
    \sgn\!\LB\ellv_{\mathrm{rr}}[i]\RB &= 1-2\bar{\uv}_{\mathrm{rr}}[i].
\end{align}

The left child input is
\begin{equation}
    \ellv'_{\mathrm{rl}} = f\!\LB\ellv_{\mathrm{rl}},\ellv_{\mathrm{rr}}\RB.
\end{equation}
Using
\begin{equation}
    \sgn(f(a,b)) = \sgn(a)\sgn(b)
\end{equation}
and
\begin{equation}
\label{eqn:xor_sign}
    (1-2(a\XOR b))(1-2b) =1-2a,
\end{equation}
we obtain
$\sgn\!\LB\ellv'_{\mathrm{rl}}\RB = 1-2\bar{\uv}_{\mathrm{rl}}$.
Therefore, by the induction hypothesis, \ac{SC} decoding of
$\ellv'_{\mathrm{rl}}$ returns
\begin{equation}
    \hat{\uv}_{\mathrm{rl}} = \bar{\uv}_{\mathrm{rl}}.
    \label{eqn:left_equality}
\end{equation}

The right-child input is then
\begin{align}
    \ellv'_{\mathrm{rr}} = g\!\LB \ellv_{\mathrm{rl}}, \ellv_{\mathrm{rr}}, \hat{\uv}_{\mathrm{rl}}
    \RB = g\!\LB \ellv_{\mathrm{rl}}, \ellv_{\mathrm{rr}}, \bar{\uv}_{\mathrm{rl}}\RB,
\end{align}
where the second equality follows from~\Cref{eqn:left_equality}.
Component-wise we then have
\begin{equation}
    \ellv'_{\mathrm{rr}}[i] = \LB 1-2\bar{\uv}_{\mathrm{rl}}[i] \RB \ellv_{\mathrm{rl}}[i] + \ellv_{\mathrm{rr}}[i]
\end{equation}
and
\begin{equation}
    \sgn\!\LB \LB 1-2\bar{\uv}_{\mathrm{rl}}[i] \RB \ellv_{\mathrm{rl}}[i] \RB = 1-2\bar{\uv}_{\mathrm{rr}}[i].
\end{equation}
Since
\begin{equation}
    \sgn\!\LB \ellv_{\mathrm{rr}}[i] \RB = 1-2\bar{\uv}_{\mathrm{rr}}[i],
\end{equation}
both terms in $\ellv'_{\mathrm{rr}}[i]$ have the same sign. Hence,
\begin{equation}
    \sgn\!\LB \ellv'_{\mathrm{rr}} \RB = 1-2\bar{\uv}_{\mathrm{rr}}.
\end{equation}
Applying the induction hypothesis again gives
\begin{equation}
    \hat{\uv}_{\mathrm{rr}} = \bar{\uv}_{\mathrm{rr}}.
\end{equation}

Finally, \ac{SC} combines the two child outputs as
\begin{align}
    \hat{\uv}_{\rtse} = \LB \hat{\uv}_{\mathrm{rl}} \XOR \hat{\uv}_{\mathrm{rr}} \RB \,||\, \hat{\uv}_{\mathrm{rr}}
    = \LB \bar{\uv}_{\mathrm{rl}} \XOR \bar{\uv}_{\mathrm{rr}} \RB \,||\, \bar{\uv}_{\mathrm{rr}}
    = \bar{\uv}_{\rtse}.
\end{align}
\end{proof}

This yields the following theorem.

\begin{theorem}
\label{theorem:spec_sc_equivalence}
For any input \ac{LLR} vector $\ellv$, \ac{Spec-SC} and conventional \ac{SC} decoding produce
identical outputs, i.e.,
\begin{equation}
    \operatorname{Spec\text{-}SC}(\ellv) = \operatorname{SC}(\ellv).
\end{equation}
\end{theorem}

\begin{proof}
    \textit{By contradiction.}
    
    Assume that there exists a node for which $\bar{\uv} \neq \hat{\uv}$, where $\bar{\uv} = (\bar{\uv}_{\ltse} \XOR \bar{\uv}_{\rtse}) || \bar{\uv}_{\rtse}$ and $ \hat{\uv} = ( \hat{\uv}_{\ltse} \XOR  \hat{\uv}_{\rtse})||  \hat{\uv}_{\rtse}$ are outputs of the \ac{Spec-SC} and \ac{SC}, respectively. However, we know that $\bar{\uv}_{\ltse} = \hat{\uv}_{\ltse}$ as the left decoding of \ac{Spec-SC} is the same as \ac{SC}. Thus, any difference between \ac{Spec-SC} and \ac{SC} can only come from the right-branch output.

    If the verification test fails, \ac{Spec-SC} discards the speculative candidate and recursively decodes the right subtree using conventional \ac{SC} decoding on $\ellv'_{\rtse}$. In that case, the right-branch outputs are identical; this means we only have to consider the case in which verification succeeds.

    Now consider the case in which the verification test succeeds. Then $\delta(\qv,\pv)=1$ and $\sgn(\ellv_{\rtse}) = \sgn(1-2\bar{\uv}_{\rtse})$. This means that the hard-decision vector associated with $\ellv_{\rtse}$ is already the candidate $\bar{\uv}_{\rtse}$. Since the polar transform at the right subtree is one-to-one, this candidate corresponds to a unique valid right-branch decision. Therefore, \ac{SC} decoding of the right branch with input \ac{LLR} $\ellv_{\rtse}$ cannot produce a different output. Equivalently, by Lemma~\ref{lemma}, \ac{SC} decoding of $\ellv_{\rtse}$ outputs $\bar{\uv}_{\rtse}$. Thus,
    $ \hat{\uv}_{\rtse} = \bar{\uv}_{\rtse}.$

    Therefore,
    \begin{equation}
        \hat{\uv}_{\ptse}
        =
        ( \hat{\uv}_{\ltse}
        \XOR
         \hat{\uv}_{\rtse})
        \,||\,
        \hat{\uv}_{\rtse}
        =
        (\bar{\uv}_{\ltse}
        \XOR
        \bar{\uv}_{\rtse})
        \,||\,
        \bar{\uv}_{\rtse}
        =
        \bar{\uv}_{\ptse}.
    \end{equation}
    This contradicts the assumption that
    $\bar{\uv}_{\ptse}\neq\uv_{\ptse}$. Hence, the
    original assumption is false, and \ac{Spec-SC} and \ac{SC} produce
    identical outputs for any input \ac{LLR} vector $\ellv$.
\end{proof}

\Cref{alg:specsc} summarizes the \ac{Spec-SC} decoding procedure as pseudocode.

\begin{algorithm}[t]
\caption{Speculative successive-cancellation decoding}
\label{alg:specsc}
\begin{algorithmic}[1]
\Require Channel \ac{LLR} vector $\ellv$, rate profile
         $\dv$, and set of leaf decoders $\mathcal{D}$
\Ensure Codeword estimate $\hat{\xv}$

\State $\hat{\xv}
       \gets \Call{Spec-SC}
       {\ellv,\dv}$
\AlgRule

\Procedure{Spec-SC}{$\ellv,\dv$}
    \State $N \gets \operatorname{length}(\ellv)$

    \If{$\dv$ matches a leaf decoder in $\mathcal{D}$}
        \State \Return
        $\Call{Decode}{\ellv,\dv}$
    \EndIf

    \Statex \hspace{\algorithmicindent}
        \textit{Execute Lines 7 and 8 concurrently:}

    \State $\hat{\uv}_{\ltse}
           \gets
           \Call{Spec-SC}
           {f(\ellv_{\ltse}, \ \ellv_{\rtse}),
            \dv_{0:\nicefrac{N}{2}}}$

    \State $\bar{\uv}_{\rtse}
           \gets
           \Call{FastSimplifiedSC}
           {\ellv_{\rtse},\dv_{\nicefrac{N}{2}:N}}$

    \State $\ellv'_{\rtse}
           \gets
           g(\ellv_{\ltse},
             \ellv_{\rtse},
             \hat{\uv}_{\ltse})$

    \If{$1 - 2\bar{\uv}_{\rtse}
         =
         \sgn(\ellv'_{\rtse})$}
        \Comment{Speculation succeeds}
        \State $\hat{\uv}_{\rtse}
               \gets \bar{\uv}_{\rtse}$
    \Else
        \Comment{Speculation fails}
        \State $\hat{\uv}_{\rtse}
               \gets
               \Call{Spec-SC}
               {\ellv'_{\rtse},
                \dv_{\nicefrac{N}{2}:N}}$
    \EndIf

    \State \Return
    $(\hat{\uv}_{\ltse}\XOR\hat{\uv}_{\rtse})
     \,\Vert\,\hat{\uv}_{\rtse}$
\EndProcedure
\AlgRule
\Procedure{FastSimplifiedSC}{$\ellv,\dv$}
    \State $N \gets \operatorname{length}(\ellv)$

    \If{$\dv$ matches a leaf decoder in $\mathcal{D}$}
        \State \Return
        $\Call{Decode}{\ellv,\dv}$
    \EndIf

    \State $\hat{\uv}_{\ltse}
           \gets
           \Call{FastSimplifiedSC}
           {f(\ellv_{\ltse}, \ \ellv_{\rtse}),
            \dv_{0:\nicefrac{N}{2}}}$

    \State $\hat{\uv}_{\rtse}
           \gets
           \Call{FastSimplifiedSC}
           {g(\ellv_{\ltse}, \ \ellv_{\rtse}, \hat{\uv}_{\ltse}),
            \dv_{\nicefrac{N}{2}:N}}$

    \State \Return
    $(\hat{\uv}_{\ltse}\XOR\hat{\uv}_{\rtse})
     \,\Vert\,\hat{\uv}_{\rtse}$
\EndProcedure

\end{algorithmic}
\end{algorithm}

Another possible approach is segmented \ac{CRC}-based verification, such as those proposed in~\cite{li2025enhanced,zhou2016segmented,hashemi2016partion, hashemi2017partition}, where this verification could be performed directly on each decoded segment. 

\begin{figure*}[!t]
    \centering
    \hspace{25pt}\ref{legend:polar_fer_4k}\\\vspace{5pt}
    \subfigure{
        \label{fig:ber_4k}
        \begin{tikzpicture}
    \pgfplotsset{
        label style = {font=\fontsize{9pt}{7.2}\selectfont},
        tick label style = {font=\fontsize{7pt}{7.2}\selectfont}
    }
    \begin{axis}[
        graph/BERFER,
        ylabel={BER}, 
        legend to name=legend:polar_ber_4k,
        legend columns=3,
        legend pos=north east,
    ]

        \addlegendimage{legend image with text={$\mathcal{P}(4096, 1024)$}}
        \addlegendentry{}
        \addlegendimage{legend image with text={$\mathcal{P}(4096, 2048)$}}
        \addlegendentry{}
        \addlegendimage{legend image with text={$\mathcal{P}(4096, 3584)$}}
        \addlegendentry{}
        
        \addplot[solid, method/SC, rate/quarter]
        table{
        1.0 1.1521e-01
        1.5 1.4940e-02
        2.0 6.6955e-04
        2.5 1.1715e-05
        3.0 5.1917e-08 
        };
        \addlegendentry{}
        
        \addplot[solid, method/SC, rate/half]
        table{
        1.0 4.5937e-01
        1.5 2.3000e-01
        2.0 2.6505e-02
        2.5 5.4295e-04
        3.0 2.9107e-06
        };
        \addlegendentry{}

        \addplot[solid, method/SC, rate/seveneighths]
        table{
        1.0 4.8270e-01
        2.0 4.7367e-01
        3.0 4.2479e-01
        4.0 1.2611e-02
        5.0 1.9786e-05
        6.0 4.7082e-08 
        };
        \addlegendentry{SC}

        \addplot[solid, method/SpecSC, rate/quarter]
        table{
        1.0 1.1866e-01
        1.5 1.5040e-02
        2.0 6.6734e-04
        2.5 1.0715e-05
        3.0 6.4101e-08 
        };
        \addlegendentry{}
        
        \addplot[solid, method/SpecSC, rate/half]
        table{
        1.0 4.5808e-01
        1.5 2.2000e-01
        2.0 2.5998e-02
        2.5 5.5295e-04
        3.0 2.9107e-06
        };
        \addlegendentry{}

        \addplot[solid, method/SpecSC, rate/seveneighths]
        table{
        1.0 4.8259e-01
        2.0 4.7354e-01
        3.0 4.2272e-01 
        4.0 1.1825e-02
        5.0 1.8226e-05
        6.0 3.3531e-08
        };
        \addlegendentry{Spec-SC}

        \addplot[solid, method/BP20, rate/quarter]
        table{
        1.0 4.2242e-02
        1.5 2.9703e-03
        2.0 1.7586e-04
        2.5 2.6282e-05
        3.0 5.2984e-06
        };
        \addlegendentry{}

        \addplot[solid, method/BP20, rate/half]
        table{
        1.0 4.1764e-01
        1.5 1.1968e-01
        2.0 5.3610e-03
        2.5 2.1784e-04
        3.0 7.7146e-06
        };
        \addlegendentry{}

        \addplot[solid, method/BP20, rate/seveneighths]
        table{
        1.0 4.7648e-01
        1.5 4.6804e-01
        2.0 4.5785e-01
        2.5 4.4612e-01
        3.0 3.8225e-01
        3.5 1.0898e-01
        4.0 7.1095e-03
        4.5 3.6470e-04
        5.0 3.6109e-05
        5.5 2.6829e-06
        6.0 7.1731e-07
        };
        \addlegendentry{BP-20}

    \end{axis}
\end{tikzpicture}
    }
    \subfigure{
        \label{fig:fer_4k}
        \begin{tikzpicture}
    \pgfplotsset{
        label style = {font=\fontsize{9pt}{7.2}\selectfont},
        tick label style = {font=\fontsize{7pt}{7.2}\selectfont}
    }
    \begin{axis}[
        graph/BERFER,
        ylabel={FER}, 
        legend to name=legend:polar_fer_4k,
        legend columns=4,
        legend pos=north east,
    ]

        \addlegendimage{empty legend}
        \addlegendentry{}
        \addlegendimage{legend image with text={$\mathcal{P}(4096, 1024)$}}
        \addlegendentry{}
        \addlegendimage{legend image with text={$\mathcal{P}(4096, 2048)$}}
        \addlegendentry{}
        \addlegendimage{legend image with text={$\mathcal{P}(4096, 3584)$}}
        \addlegendentry{}

        \addlegendimage{empty legend}
        \addlegendentry{SC}
        \addplot[solid, method/SC, rate/quarter]
        table{
        1.0 3.4192e-01
        1.5 5.1130e-02
        2.0 2.6855e-03 
        2.5 4.9877e-05
        3.0 2.9902e-07 
        };
        \addlegendentry{}
        
        \addplot[solid, method/SC, rate/half]
        table{
        1.0 9.8438e-01
        1.5 5.7129e-01
        2.0 8.3984e-02
        2.5 2.3717e-03
        3.0 1.5259e-05 
        };
        \addlegendentry{}

        \addplot[solid, method/SC, rate/seveneighths]
        table{
        1.0 1.0000e+00
        1.5 1.0000e+00
        2.0 1.0000e+00
        2.5 1.0000e+00
        3.0 9.8828e-01
        3.5 5.9766e-01
        4.0 8.0566e-02
        4.5 4.3733e-03
        5.0 3.2128e-04
        5.5 1.9338e-05
        6.0 1.0821e-06
        };
        \addlegendentry{}

        \addlegendimage{empty legend}
        \addlegendentry{Spec-SC}
        \addplot[solid, method/SpecSC, rate/quarter]
        table{
        1.0 3.2617e-01
        1.5 5.0130e-02
        2.0 2.6470e-03
        2.5 4.8877e-05
        3.0 3.2963e-07 
        };
        \addlegendentry{}
        
        \addplot[solid, method/SpecSC, rate/half]
        table{
        1.0 9.8215e-01
        1.5 5.5129e-01
        2.0 7.9926e-02
        2.5 2.4717e-03
        3.0 1.5259e-05
        };
        \addlegendentry{}

        \addplot[solid, method/SpecSC, rate/seveneighths]
        table{
        1.0 1.0000e+00
        1.5 1.0000e+00
        2.0 1.0000e+00
        2.5 1.0000e+00
        3.0 9.9084e-01
        3.5 5.8766e-01
        4.0 7.4585e-02
        4.5 4.5733e-03
        5.0 3.6905e-04 
        5.5 2.0338e-05
        6.0 9.0054e-07
        };
        \addlegendentry{}

        \addlegendimage{empty legend}
        \addlegendentry{BP-20}
        \addplot[solid, method/BP20, rate/quarter]
        table{
        1.0 1.5723e-01
        1.5 1.2939e-02
        2.0 9.8643e-04
        2.5 2.0345e-04
        3.0 9.1524e-05
        };
        \addlegendentry{}

        \addplot[solid, method/BP20, rate/half]
        table{
        1.0 9.8535e-01
        1.5 4.2676e-01
        2.0 2.5391e-02
        2.5 1.4944e-03
        3.0 2.4232e-04
        };
        \addlegendentry{}

        \addplot[solid, method/BP20, rate/seveneighths]
        table{
        1.0 1.0000e+00
        1.5 1.0000e+00
        2.0 1.0000e+00
        2.5 1.0000e+00
        3.0 9.9219e-01
        3.5 5.2246e-01
        4.0 6.9336e-02
        4.5 6.7708e-03
        5.0 7.3426e-04
        5.5 1.0422e-04 
        6.0 2.4630e-05
        };
        \addlegendentry{}

    \end{axis}
\end{tikzpicture}
    }
    \caption{BER and FER performance comparison between the \ac{SC} decoder, \ac{BP} decoder and the proposed \ac{Spec-SC} decoder for polar codes of length $N=4096$ with code rates $R=\nicefrac{1}{4}$, $R=\nicefrac{1}{2}$, and $R=\nicefrac{7}{8}$, corresponding to $\Pc(4096,1024)$, $\Pc(4096,2048)$, and $\Pc(4096,3584)$, respectively.}
    \label{fig:4k_ber_fer}
\end{figure*}

\section{Numerical Results}
\label{sec:results}
This section presents simulation results evaluating the performance of the proposed \ac{Spec-SC} decoder in terms of error-correction capability and average number of nodes traversed in the decoding tree. A cycle-level model is also included.

\subsection{Simulation Setup}
\label{subsec:sim_setup}
All experiments are conducted using the Sionna framework~\cite{sionna} to simulate a \ac{BPSK} transmission over an \ac{AWGN} channel. The \ac{AWGN} channel provides a common reference for evaluating the decoding algorithms independently of a particular communication system or physical-layer implementation.

The polar codes used in the experiments are constructed based on the Gaussian approximation method~\cite{Trifonov2012Efficient} at an $\EbNo$ of 3\,dB. Codes of length $N=4096$ are evaluated for code rates $R \in \{\nicefrac{1}{4}, \nicefrac{1}{2}, \nicefrac{7}{8}\}$. Simulations are terminated once the frame error rate (FER) performance falls below a threshold of $10^{-4}$ with at least 256 observed errors. 

\subsection{Error Correction Performance}
\label{subsec:ber_fer_performance}
\Cref{fig:4k_ber_fer} compares the BER and FER performance of \ac{SC} and \ac{Spec-SC} decoding as well as \ac{BP} decoding with a maximum of 20 iterations. As predicted by \Cref{theorem:spec_sc_equivalence}, the \ac{Spec-SC} decoder achieves identical performance to the conventional \ac{SC} decoder across all $\EbNo$ values. This demonstrates that the speculative mechanism introduced in \ac{Spec-SC} does not compromise decoding accuracy.

The \ac{BP} decoder achieves a lower \ac{FER} than \ac{SC} and \ac{Spec-SC} from 1\,dB to ${\sim}2.2$\,dB for $\Pc(4096,1024)$ and 1\,dB to ${\sim}2.7$\,dB for $\Pc(4096,2048)$. Beyond these ranges, the \ac{BP} decoder \ac{FER} becomes worse than that of \ac{SC} and \ac{Spec-SC}. For $\Pc(4096,3584)$, the \ac{FER} of the \ac{BP} decoder closely matches that of \ac{SC} and \ac{Spec-SC} from 1\,dB to ${\sim}4.6$\,dB, after which the \ac{BP} decoder's performance begins to diverge. This is consistent with the error floor behavior reported for \ac{BP} decoding of polar codes.

\begin{figure}[!t]
    \centering
    \hspace{25pt}\ref{legend:polar_nodes_traversed}\\\vspace{5pt}
    \begin{tikzpicture}
    \pgfplotsset{
        label style = {font=\fontsize{9pt}{7.2}\selectfont},
        tick label style = {font=\fontsize{7pt}{7.2}\selectfont}
    }
    \begin{axis}[
        width=0.9\columnwidth,
        height=0.75\columnwidth,
        xlabel={$E_\mathrm{b}/N_0$ (dB)}, 
        ylabel={Average Sequential Nodes Traversed}, 
        ylabel style={yshift=-7pt},
        grid=both,
        grid style=dashed,
        xmajorgrids=true,
        ymajorgrids=true,
        minor x tick num=1,
        legend style={
            font=\fontsize{7pt}{7.2}\selectfont,
            cells={anchor=west},
            column sep=2mm,
            legend columns=4,
            at={(0.5,1.15)},
            anchor=south
        },
        legend to name=legend:polar_nodes_traversed,
    ]

        \addlegendimage{empty legend}
        \addlegendentry{}
        \addlegendimage{legend image with text={$\mathcal{P}(4096, 1024)$}}
        \addlegendentry{}
        \addlegendimage{legend image with text={$\mathcal{P}(4096, 2048)$}}
        \addlegendentry{}
        \addlegendimage{legend image with text={$\mathcal{P}(4096, 3584)$}}
        \addlegendentry{}

        \addlegendimage{empty legend}
        \addlegendentry{SC}
        \addplot[solid, method/SC, rate/quarter]
        table{
        1.0 383.0
        1.5 383.0
        2.0 383.0
        2.5 383.0
        3.0 383.0
        3.5 383.0
        4.0 383.0
        4.5 383.0
        5.0 383.0
        5.5 383.0
        6.0 383.0
        };
        \addlegendentry{}

        \addplot[solid, method/SC, rate/half]
        table{
        1.0 467.0
        1.5 467.0
        2.0 467.0
        2.5 467.0
        3.0 467.0
        3.5 467.0
        4.0 467.0
        4.5 467.0
        5.0 467.0
        5.5 467.0
        6.0 467.0
        };
        \addlegendentry{}

        \addplot[solid, method/SC, rate/seveneighths]
        table{
        1.0 197.0
        1.5 197.0
        2.0 197.0
        2.5 197.0
        3.0 197.0
        3.5 197.0
        4.0 197.0
        4.5 197.0
        5.0 197.0
        5.5 197.0
        6.0 197.0
        };
        \addlegendentry{}

        \addlegendimage{empty legend}
        \addlegendentry{Spec-SC}
        \addplot[solid, method/SpecSC, rate/quarter]
        table{
        1.0 231.5
        1.5 170.2
        2.0 137.2
        2.5 115.8
        3.0 98.1
        3.5 83.8
        4.0 71.3
        4.5 61.3
        5.0 52.6
        5.5 45.4
        6.0 39.4    
        };
        \addlegendentry{}
        
        \addplot[solid, method/SpecSC, rate/half]
        table{
        1.0 419.5
        1.5 321.0
        2.0 197.8
        2.5 143.5
        3.0 112.7
        3.5 89.9
        4.0 72.8
        4.5 59.0
        5.0 47.5
        5.5 37.8
        6.0 30.4
        };
        \addlegendentry{}

        \addplot[solid, method/SpecSC, rate/seveneighths]
        table{
        1.0 183.7
        1.5 183.4
        2.0 183.4
        2.5 183.2
        3.0 179.8
        3.5 147.2
        4.0 99.9
        4.5 66.9
        5.0 46.7
        5.5 34.4
        6.0 21.5
        };
        \addlegendentry{}
        
    \end{axis}
\end{tikzpicture}
    \caption{Average number of nodes traversed in the main decoding tree. The count excludes the subsequent traversal of successfully speculated subtrees, which is skipped when verification succeeds.}
    \label{fig:nodes_traversed}
\end{figure}

\subsection{Average Sequential Nodes Traversed}
\label{subsec:nodes_traversed}
\Cref{fig:nodes_traversed} illustrates the average number of sequential nodes traversed in the main decoding tree for both \ac{SC} and \ac{Spec-SC} decoders. The results are obtained by averaging over 1000 frames per simulation point. As $\EbNo$ increases, the number of nodes traversed in the decoding tree decreases, indicating that fewer re-decodings are required at higher channel reliability.

The operating \ac{FER} requirement depends on the target application. We first consider an \ac{FER} of $10^{-2}$ as a reference point for evaluating the reduction in sequential node traversal. The \ac{Spec-SC} decoder demonstrates substantial reductions compared with conventional \ac{SC} decoding. 

For $\Pc(4096,1024)$, the target FER of $10^{-2}$ is reached at 2\,dB, where the average number of traversed nodes is reduced from 383 to 137.2, corresponding to a $64\%$ reduction.

For $\Pc(4096,2048)$, the same target FER is achieved at 2.5\,dB, with the average number of traversed nodes reduced from 467 to 143.5, corresponding to a $69\%$ reduction.

Finally, for $\Pc(4096,3584)$, the FER target is met at 4.5\,dB, where the average number of traversed nodes decreases from  197 to 66.9, corresponding to a $66\%$ reduction.

While \ac{Spec-SC} introduces additional speculative computations, it shortens the main decoding path, enabling large subtrees to be skipped when verification succeeds. As a result, even with extra speculative effort, \ac{Spec-SC} can improve throughput on architectures that support parallel execution or dedicated speculative units. The actual hardware performance, however, depends on implementation details such as pipeline depth and parallelism. 

\subsection{Ablation Study on Verification Success Rate}
\label{sec:study_success_rate}

\begin{figure*}[t]
    \centering
    \subfigure[$\EbNo = 1.0$\,dB]{
        \label{fig:success_prob_1.0_dB}
        \begin{tikzpicture}
    \pgfplotsset{
        label style = {font=\fontsize{9pt}{7.2}\selectfont},
        tick label style = {font=\fontsize{7pt}{7.2}\selectfont},
    }

    \begin{axis}[
        graph/SuccessRate,
    ]

    \addplot+[
        only marks,
        mark=*,
        mark size=1.8,
        opacity=0.85,
        color=matplottikz-color3,
        mark options={solid, fill=matplottikz-color3}
    ]
    table[row sep=\\] {
           0 0.00 \\
           1 0.00 \\
           2 0.00 \\
           3 0.00 \\
           4 0.00 \\
           5 0.00 \\
           7 0.00 \\
           8 0.00 \\
           9 0.00 \\
          10 0.45 \\
          11 0.00 \\
          15 0.00 \\
          16 0.00 \\
          17 0.00 \\
          18 0.12 \\
          19 0.00 \\
          20 0.34 \\
          21 0.45 \\
          23 0.00 \\
          24 0.45 \\
          31 0.00 \\
          32 0.00 \\
          33 0.00 \\
          34 0.02 \\
          35 0.00 \\
          36 0.07 \\
          37 0.12 \\
          39 0.00 \\
          40 0.25 \\
          41 0.32 \\
          43 0.47 \\
          47 0.00 \\
          48 0.28 \\
          49 0.45 \\
          64 0.00 \\
          65 0.00 \\
          66 0.00 \\
          67 0.00 \\
          68 0.00 \\
          69 0.01 \\
          71 0.00 \\
          72 0.02 \\
          73 0.10 \\
          75 0.11 \\
          79 0.00 \\
          80 0.17 \\
          81 0.21 \\
          83 0.33 \\
          87 0.46 \\
          95 0.00 \\
          96 0.29 \\
          97 0.29 \\
          99 0.45 \\
         130 0.00 \\
         132 0.00 \\
         133 0.00 \\
         134 0.15 \\
         136 0.00 \\
         137 0.00 \\
         138 0.30 \\
         139 0.01 \\
         144 0.01 \\
         145 0.03 \\
         147 0.09 \\
         151 0.14 \\
         159 0.00 \\
         160 0.09 \\
         161 0.18 \\
         163 0.23 \\
         167 0.33 \\
         175 0.46 \\
         191 0.00 \\
         192 0.15 \\
         193 0.24 \\
         195 0.24 \\
         199 0.45 \\
         262 0.03 \\
         266 0.03 \\
         268 0.07 \\
         269 0.17 \\
         274 0.14 \\
         276 0.22 \\
         277 0.29 \\
         279 0.03 \\
         280 0.29 \\
         289 0.04 \\
         290 0.28 \\
         291 0.04 \\
         292 0.47 \\
         295 0.09 \\
         303 0.18 \\
         320 0.05 \\
         321 0.10 \\
         323 0.22 \\
         327 0.29 \\
         335 0.29 \\
         351 0.44 \\
         384 0.07 \\
         385 0.17 \\
         387 0.32 \\
         391 0.30 \\
         399 0.48 \\
         526 0.24 \\
         534 0.33 \\
         538 0.43 \\
         553 0.26 \\
         555 0.35 \\
         560 0.33 \\
         561 0.35 \\
         580 0.31 \\
         581 0.29 \\
         584 0.31 \\
         585 0.42 \\
         642 0.32 \\
         647 0.26 \\
         655 0.35 \\
         671 0.33 \\
         703 0.44 \\
         770 0.43 \\
         775 0.36 \\
         783 0.33 \\
         799 0.44 \\
    };

    \end{axis}
\end{tikzpicture}
    }
    \subfigure[$\EbNo = 2.0$\,dB]{
        \label{fig:success_prob_2.0_dB}
        \begin{tikzpicture}
    \pgfplotsset{
        label style = {font=\fontsize{9pt}{7.2}\selectfont},
        tick label style = {font=\fontsize{7pt}{7.2}\selectfont},
    }

    \begin{axis}[
        graph/SuccessRate,
    ]

    \addplot+[
        only marks,
        mark=*,
        mark size=1.8,
        opacity=0.85,
        color=matplottikz-color3,
        mark options={solid, fill=matplottikz-color3}
    ]
    table[row sep=\\] {
           0 0.00 \\
           1 0.00 \\
           2 0.00 \\
           3 0.00 \\
           4 0.00 \\
           5 0.00 \\
           7 0.00 \\
           8 0.00 \\
           9 0.00 \\
          10 0.43 \\
          11 0.00 \\
          15 0.00 \\
          16 0.00 \\
          17 0.00 \\
          18 0.14 \\
          19 0.00 \\
          20 0.36 \\
          21 0.49 \\
          23 0.00 \\
          24 0.49 \\
          31 0.00 \\
          32 0.00 \\
          33 0.00 \\
          34 0.01 \\
          35 0.00 \\
          36 0.07 \\
          37 0.11 \\
          39 0.00 \\
          40 0.24 \\
          41 0.36 \\
          43 0.46 \\
          47 0.00 \\
          48 0.31 \\
          49 0.45 \\
          64 0.00 \\
          65 0.00 \\
          66 0.00 \\
          67 0.00 \\
          68 0.00 \\
          69 0.01 \\
          71 0.00 \\
          72 0.02 \\
          73 0.10 \\
          75 0.12 \\
          79 0.00 \\
          80 0.17 \\
          81 0.21 \\
          83 0.34 \\
          87 0.47 \\
          95 0.00 \\
          96 0.27 \\
          97 0.32 \\
          99 0.47 \\
         130 0.00 \\
         132 0.00 \\
         133 0.00 \\
         134 0.14 \\
         136 0.00 \\
         137 0.01 \\
         138 0.34 \\
         139 0.01 \\
         144 0.01 \\
         145 0.04 \\
         147 0.09 \\
         151 0.15 \\
         159 0.00 \\
         160 0.11 \\
         161 0.17 \\
         163 0.26 \\
         167 0.34 \\
         175 0.46 \\
         191 0.00 \\
         192 0.15 \\
         193 0.24 \\
         195 0.22 \\
         199 0.47 \\
         262 0.03 \\
         266 0.06 \\
         268 0.07 \\
         269 0.18 \\
         274 0.13 \\
         276 0.23 \\
         277 0.28 \\
         279 0.04 \\
         280 0.26 \\
         289 0.03 \\
         290 0.28 \\
         291 0.06 \\
         292 0.40 \\
         295 0.11 \\
         303 0.15 \\
         320 0.05 \\
         321 0.10 \\
         323 0.22 \\
         327 0.27 \\
         335 0.31 \\
         351 0.49 \\
         384 0.07 \\
         385 0.17 \\
         387 0.27 \\
         391 0.28 \\
         399 0.46 \\
         526 0.26 \\
         534 0.31 \\
         538 0.43 \\
         553 0.24 \\
         555 0.35 \\
         560 0.30 \\
         561 0.33 \\
         580 0.32 \\
         581 0.34 \\
         584 0.34 \\
         585 0.43 \\
         642 0.31 \\
         647 0.25 \\
         655 0.34 \\
         671 0.36 \\
         703 0.46 \\
         770 0.42 \\
         775 0.34 \\
         783 0.35 \\
         799 0.47 \\
    };

    \end{axis}
\end{tikzpicture}
    }
    \subfigure[$\EbNo = 3.0$\,dB]{
        \label{fig:success_prob_3.0_dB}
        \begin{tikzpicture}
    \pgfplotsset{
        label style = {font=\fontsize{9pt}{7.2}\selectfont},
        tick label style = {font=\fontsize{7pt}{7.2}\selectfont},
    }

    \begin{axis}[
        graph/SuccessRate,
    ]

    \addplot+[
        only marks,
        mark=*,
        mark size=1.8,
        opacity=0.85,
        color=matplottikz-color3,
        mark options={solid, fill=matplottikz-color3}
    ]
    table[row sep=\\] {
   0 0.00 \\
   1 0.00 \\
   2 0.00 \\
   3 0.00 \\
   4 0.00 \\
   5 0.00 \\
   7 0.00 \\
   8 0.00 \\
   9 0.00 \\
  10 0.44 \\
  11 0.00 \\
  15 0.00 \\
  16 0.00 \\
  17 0.00 \\
  18 0.15 \\
  19 0.00 \\
  20 0.35 \\
  21 0.45 \\
  23 0.00 \\
  24 0.46 \\
  31 0.00 \\
  32 0.00 \\
  33 0.00 \\
  34 0.03 \\
  35 0.00 \\
  36 0.10 \\
  37 0.15 \\
  39 0.00 \\
  40 0.25 \\
  41 0.34 \\
  43 0.46 \\
  47 0.00 \\
  48 0.28 \\
  49 0.47 \\
  64 0.00 \\
  65 0.00 \\
  66 0.02 \\
  67 0.00 \\
  68 0.03 \\
  69 0.05 \\
  71 0.00 \\
  72 0.04 \\
  73 0.11 \\
  75 0.18 \\
  79 0.00 \\
  80 0.18 \\
  81 0.21 \\
  83 0.34 \\
  87 0.48 \\
  95 0.00 \\
  96 0.27 \\
  97 0.29 \\
  99 0.45 \\
 130 0.00 \\
 132 0.01 \\
 133 0.03 \\
 134 0.30 \\
 136 0.03 \\
 137 0.04 \\
 138 0.38 \\
 139 0.06 \\
 144 0.04 \\
 145 0.05 \\
 147 0.12 \\
 151 0.18 \\
 159 0.00 \\
 160 0.11 \\
 161 0.20 \\
 163 0.25 \\
 167 0.36 \\
 175 0.46 \\
 191 0.01 \\
 192 0.16 \\
 193 0.24 \\
 195 0.26 \\
 199 0.46 \\
 262 0.12 \\
 266 0.22 \\
 268 0.25 \\
 269 0.31 \\
 274 0.24 \\
 276 0.31 \\
 277 0.37 \\
 279 0.08 \\
 280 0.34 \\
 289 0.06 \\
 290 0.34 \\
 291 0.08 \\
 292 0.45 \\
 295 0.14 \\
 303 0.21 \\
 320 0.08 \\
 321 0.11 \\
 323 0.25 \\
 327 0.30 \\
 335 0.29 \\
 351 0.45 \\
 384 0.09 \\
 385 0.17 \\
 387 0.28 \\
 391 0.30 \\
 399 0.46 \\
 526 0.40 \\
 534 0.47 \\
 538 0.49 \\
 553 0.33 \\
 555 0.43 \\
 560 0.37 \\
 561 0.40 \\
 580 0.34 \\
 581 0.38 \\
 584 0.33 \\
 585 0.44 \\
 642 0.35 \\
 647 0.24 \\
 655 0.35 \\
 671 0.34 \\
 703 0.48 \\
 770 0.45 \\
 775 0.34 \\
 783 0.35 \\
 799 0.44 \\
    };

    \end{axis}
\end{tikzpicture}
    }
    \subfigure[$\EbNo = 4.0$\,dB]{
        \label{fig:success_prob_1.0_dB}
        \begin{tikzpicture}
    \pgfplotsset{
        label style = {font=\fontsize{9pt}{7.2}\selectfont},
        tick label style = {font=\fontsize{7pt}{7.2}\selectfont},
    }

    \begin{axis}[
        graph/SuccessRate,
    ]

    \addplot+[
        only marks,
        mark=*,
        mark size=1.8,
        opacity=0.85,
        color=matplottikz-color3,
        mark options={solid, fill=matplottikz-color3}
    ]
    table[row sep=\\] {
   0 0.00 \\
   1 0.00 \\
   2 0.28 \\
   3 0.00 \\
   4 0.17 \\
   5 0.30 \\
   7 0.00 \\
   8 0.14 \\
   9 0.27 \\
  10 0.92 \\
  11 0.45 \\
  15 0.00 \\
  16 0.08 \\
  17 0.20 \\
  18 0.76 \\
  19 0.34 \\
  20 0.86 \\
  21 0.92 \\
  23 0.64 \\
  24 0.90 \\
  31 0.00 \\
  32 0.02 \\
  33 0.10 \\
  34 0.70 \\
  35 0.23 \\
  36 0.81 \\
  37 0.89 \\
  39 0.34 \\
  40 0.79 \\
  41 0.88 \\
  43 0.93 \\
  47 0.60 \\
  48 0.78 \\
  49 0.91 \\
  64 0.00 \\
  65 0.01 \\
  66 0.56 \\
  67 0.04 \\
  68 0.72 \\
  69 0.67 \\
  71 0.16 \\
  72 0.70 \\
  73 0.81 \\
  75 0.89 \\
  79 0.38 \\
  80 0.85 \\
  81 0.90 \\
  83 0.86 \\
  87 0.93 \\
  95 0.62 \\
  96 0.70 \\
  97 0.76 \\
  99 0.92 \\
 130 0.17 \\
 132 0.31 \\
 133 0.50 \\
 134 0.82 \\
 136 0.56 \\
 137 0.71 \\
 138 0.69 \\
 139 0.83 \\
 144 0.75 \\
 145 0.79 \\
 147 0.81 \\
 151 0.89 \\
 159 0.34 \\
 160 0.73 \\
 161 0.84 \\
 163 0.91 \\
 167 0.86 \\
 175 0.93 \\
 191 0.58 \\
 192 0.78 \\
 193 0.87 \\
 195 0.87 \\
 199 0.90 \\
 262 0.67 \\
 266 0.81 \\
 268 0.84 \\
 269 0.82 \\
 274 0.74 \\
 276 0.83 \\
 277 0.93 \\
 279 0.77 \\
 280 0.89 \\
 289 0.68 \\
 290 0.91 \\
 291 0.79 \\
 292 0.71 \\
 295 0.83 \\
 303 0.89 \\
 320 0.75 \\
 321 0.79 \\
 323 0.84 \\
 327 0.90 \\
 335 0.93 \\
 351 0.92 \\
 384 0.83 \\
 385 0.78 \\
 387 0.89 \\
 391 0.87 \\
 399 0.92 \\
 526 0.82 \\
 534 0.87 \\
 538 0.57 \\
 553 0.87 \\
 555 0.94 \\
 560 0.79 \\
 561 0.88 \\
 580 0.82 \\
 581 0.89 \\
 584 0.85 \\
 585 0.77 \\
 642 0.89 \\
 647 0.84 \\
 655 0.89 \\
 671 0.93 \\
 703 0.92 \\
 770 0.58 \\
 775 0.90 \\
 783 0.91 \\
 799 0.91 \\
    };

    \end{axis}
\end{tikzpicture}
    }
    \subfigure[$\EbNo = 5.0$\,dB]{
        \label{fig:success_prob_2.0_dB}
        \begin{tikzpicture}
    \pgfplotsset{
        label style = {font=\fontsize{9pt}{7.2}\selectfont},
        tick label style = {font=\fontsize{7pt}{7.2}\selectfont},
    }

    \begin{axis}[
        graph/SuccessRate,
    ]

    \addplot+[
        only marks,
        mark=*,
        mark size=1.8,
        opacity=0.85,
        color=matplottikz-color3,
        mark options={solid, fill=matplottikz-color3}
    ]
    table[row sep=\\] {
   0 0.01 \\
   1 0.01 \\
   2 0.80 \\
   3 0.04 \\
   4 0.74 \\
   5 0.80 \\
   7 0.01 \\
   8 0.81 \\
   9 0.85 \\
  10 0.98 \\
  11 0.91 \\
  15 0.00 \\
  16 0.76 \\
  17 0.84 \\
  18 0.82 \\
  19 0.89 \\
  20 0.96 \\
  21 0.98 \\
  23 0.96 \\
  24 1.00 \\
  31 0.00 \\
  32 0.45 \\
  33 0.69 \\
  34 0.94 \\
  35 0.85 \\
  36 0.97 \\
  37 0.99 \\
  39 0.91 \\
  40 0.88 \\
  41 0.98 \\
  43 0.98 \\
  47 0.96 \\
  48 0.90 \\
  49 1.00 \\
  64 0.14 \\
  65 0.31 \\
  66 0.95 \\
  67 0.54 \\
  68 0.94 \\
  69 0.93 \\
  71 0.77 \\
  72 0.94 \\
  73 0.98 \\
  75 1.00 \\
  79 0.88 \\
  80 0.99 \\
  81 1.00 \\
  83 0.98 \\
  87 0.98 \\
  95 0.96 \\
  96 0.70 \\
  97 0.88 \\
  99 1.00 \\
 130 0.79 \\
 132 0.87 \\
 133 0.92 \\
 134 1.00 \\
 136 0.94 \\
 137 0.98 \\
 138 0.76 \\
 139 0.98 \\
 144 0.98 \\
 145 0.98 \\
 147 0.98 \\
 151 1.00 \\
 159 0.83 \\
 160 0.99 \\
 161 0.94 \\
 163 0.99 \\
 167 0.97 \\
 175 0.98 \\
 191 0.94 \\
 192 0.97 \\
 193 1.00 \\
 195 1.00 \\
 199 1.00 \\
 262 0.96 \\
 266 0.99 \\
 268 0.98 \\
 269 1.00 \\
 274 1.00 \\
 276 1.00 \\
 277 1.00 \\
 279 0.97 \\
 280 1.00 \\
 289 0.94 \\
 290 1.00 \\
 291 0.96 \\
 292 0.33 \\
 295 0.99 \\
 303 0.98 \\
 320 0.98 \\
 321 0.97 \\
 323 0.98 \\
 327 1.00 \\
 335 1.00 \\
 351 0.97 \\
 384 1.00 \\
 385 0.92 \\
 387 0.95 \\
 391 1.00 \\
 399 1.00 \\
 526 1.00 \\
 534 1.00 \\
 538 1.00 \\
 553 1.00 \\
 555 1.00 \\
 560 1.00 \\
 561 1.00 \\
 580 1.00 \\
 581 1.00 \\
 584 1.00 \\
 585 1.00 \\
 642 1.00 \\
 647 0.94 \\
 655 0.99 \\
 671 1.00 \\
 703 0.99 \\
 775 0.97 \\
 783 1.00 \\
 799 1.00 \\
    };

    \end{axis}
\end{tikzpicture}
    }
    \subfigure[$\EbNo = 6.0$\,dB]{
        \label{fig:success_prob_3.0_dB}
        \begin{tikzpicture}
    \pgfplotsset{
        label style = {font=\fontsize{9pt}{7.2}\selectfont},
        tick label style = {font=\fontsize{7pt}{7.2}\selectfont},
    }

    \begin{axis}[
        graph/SuccessRate,
    ]

    \addplot+[
        only marks,
        mark=*,
        mark size=1.8,
        opacity=0.85,
        color=matplottikz-color3,
        mark options={solid, fill=matplottikz-color3}
    ]
    table[row sep=\\] {
   0 0.55 \\
   1 0.57 \\
   2 0.96 \\
   3 0.52 \\
   4 0.97 \\
   5 0.96 \\
   7 0.30 \\
   8 0.98 \\
   9 0.99 \\
  10 1.00 \\
  11 0.98 \\
  15 0.15 \\
  16 0.97 \\
  17 0.99 \\
  18 0.88 \\
  19 0.99 \\
  20 0.83 \\
  21 1.00 \\
  23 1.00 \\
  24 1.00 \\
  31 0.05 \\
  32 0.93 \\
  33 0.96 \\
  34 1.00 \\
  35 0.98 \\
  36 1.00 \\
  37 1.00 \\
  39 1.00 \\
  40 0.50 \\
  41 1.00 \\
  43 1.00 \\
  47 1.00 \\
  48 0.50 \\
  49 1.00 \\
  64 0.73 \\
  65 0.85 \\
  66 1.00 \\
  67 0.93 \\
  68 1.00 \\
  69 1.00 \\
  71 0.96 \\
  72 1.00 \\
  73 1.00 \\
  75 1.00 \\
  79 0.99 \\
  80 1.00 \\
  81 1.00 \\
  83 0.83 \\
  87 1.00 \\
  95 0.99 \\
  96 1.00 \\
  97 1.00 \\
  99 1.00 \\
 130 0.98 \\
 132 0.98 \\
 133 1.00 \\
 136 0.98 \\
 137 1.00 \\
 139 1.00 \\
 144 1.00 \\
 145 1.00 \\
 147 1.00 \\
 151 1.00 \\
 159 0.98 \\
 160 1.00 \\
 161 1.00 \\
 163 1.00 \\
 167 1.00 \\
 175 1.00 \\
 191 1.00 \\
 192 1.00 \\
 193 1.00 \\
 195 1.00 \\
 199 1.00 \\
 262 1.00 \\
 266 1.00 \\
 274 1.00 \\
 279 1.00 \\
 289 0.95 \\
 291 1.00 \\
 295 1.00 \\
 303 1.00 \\
 320 1.00 \\
 321 1.00 \\
 323 1.00 \\
 327 1.00 \\
 335 1.00 \\
 351 1.00 \\
 384 1.00 \\
 385 1.00 \\
 387 1.00 \\
 391 1.00 \\
 399 1.00 \\
 580 1.00 \\
 647 1.00 \\
 655 1.00 \\
 671 1.00 \\
 703 1.00 \\
 775 1.00 \\
 783 1.00 \\
 799 1.00 \\
    };

    \end{axis}
\end{tikzpicture}
    }
    
    \caption{Verification success rate versus node ID for $\Pc(4096, 3584)$ at different $\EbNo$ values. Node IDs follow a breadth-first labeling of the \ac{SC} decoding tree, with node 0 denoting the root and larger node IDs corresponding to lower-level subtrees closer to the leaves. Each point represents the fraction of simulated frames for which the speculative decision at that node passes the verification test. At lower $\EbNo$, accepted speculation is limited to fewer nodes because the tentative right-branch decisions are less reliable. As $\EbNo$ increases, more nodes exhibit high success rates, indicating that speculation becomes more reliable under improved channel conditions.}
    \label{fig:spec_success_rate}
\end{figure*}

\begin{figure*}[!t]
    \centering
    \begin{tikzpicture}[
    every node/.style={
        rounded corners,
        inner sep=1.2pt,
        font=\fontsize{0.6em}{0.72em}\selectfont,
        align=center
    },
    edge from parent/.style={draw},
    edge from parent path={
        (\tikzparentnode.south) -- (\tikzchildnode.north)
    },
    level distance=11mm,
    level 1/.style={sibling distance=78mm},
    level 2/.style={sibling distance=40mm},
    level 3/.style={sibling distance=20mm},
    level 4/.style={sibling distance=10mm},
    snr/.style={
        midway,
        sloped,
        above=1pt,
        draw=none,
        fill=white,
        inner sep=0.5pt,
        font=\tiny
    }
]

\node {(4096,3584)}
    child {
        node {(2048,1586)}
        child {
            node {(1024,622)}
            child {
                node {(512,195)}
                child {
                    node {(256,39)}
                }
                child {
                    node {(256,156)}
                    edge from parent node[snr] {${>}5$ dB}
                }
            }
            child {
                node {(512,427)}
                child {
                    node {(256,181)}
                }
                child {
                    node {(256,246)}
                    edge from parent node[snr] {$4.52$ dB}
                }
                edge from parent node[snr] {${>}5$ dB}
            }
        }
        child {
            node {(1024,964)}
            child {
                node {(512,455)}
                child {
                    node {(256,204)}
                }
                child {
                    node {(256,251)}
                    edge from parent node[snr] {$4.39$ dB}
                }
            }
            child {
                node {(512,509)}
                edge from parent node[snr] {$4.57$ dB}
            }
            edge from parent node[snr] {${>}5$ dB}
        }
    }
    child {
        node {(2048,1998)}
        child {
            node {(1024,975)}
            child {
                node {(512,464)}
                child {
                    node {(256,211)}
                }
                child {
                    node {(256,253)}
                    edge from parent node[snr] {$4.12$ dB}
                }
            }
            child {
                node {(512,511)}
                edge from parent node[snr] {$4.40$ dB}
            }
        }
        child {
            node {(1024,1023)}
            edge from parent node[snr] {$4.39$ dB}
        }
        edge from parent node[snr] {${>}5$ dB}
    };

\end{tikzpicture}
    \caption{Speculation-threshold tree for $\Pc(4096,3584)$. Each node is labeled by its corresponding subtree size and number of information bits, $(N_v,K_v)$. An edge label gives the minimum $\EbNo$ at which speculation along that branch succeeds in more than $50\%$ of simulated frames.}
    \label{fig:spec_threshold_tree}
\end{figure*}

To understand where speculation works best in the \ac{SC} decoding tree, we measure the verification success rate at each node across different $\EbNo$ values for $\Pc(4096, 3584)$. Here, a speculative attempt is accepted if it passes the verification condition. However, some rejected attempts may still have produced the same final output as conventional \ac{SC} decoding.

The nodes are labeled in breadth-first order, i.e., the root node is labeled $0$, its two children are labeled $1$ and $2$, their children $3,4,5,6$, and so on, from left to right. With this labeling, smaller node IDs represent larger subtrees closer to the root, while larger node IDs represent smaller subtrees closer to the leaves.

Fig.~\ref{fig:spec_success_rate} shows the verification success rate for each node at different $\EbNo$ values. At low $\EbNo$, only a small number of nodes can be speculated successfully, and the success rate varies across the tree. This is because the channel observations are noisy, so the right-branch decision is less reliable when it is made before the left branch partial sums are known. In these cases, the verification condition often fails, and the decoder reverts to conventional \ac{SC} decoding.

As $\EbNo$ increases, speculation becomes more reliable. Many nodes show a much higher success rate, and some nodes have success rates close to one. This means that, for these nodes, the speculative decision is almost always the same as the decision produced by conventional \ac{SC} decoding. Therefore, these subtrees can often be decoded speculatively without changing the final decoding result.

To further summarize this behavior, Fig.~\ref{fig:spec_threshold_tree} shows the first $\EbNo$ value at which each speculative branch achieves a success rate greater than $50\%$. Each node is labeled by the corresponding subtree size and number of information bits, $(N_v,K_v)$, while each edge label gives the minimum $\EbNo$ required for the verification success rate on that branch to exceed this threshold. Therefore, for a given operating $\EbNo$, speculation can be selectively enabled only on branches whose labeled threshold is below that $\EbNo$. For example, at $\EbNo=4.5$\,dB, branches labeled $4.12$\,dB, $4.39$\,dB, and $4.40$\,dB would be selected for speculation, while branches labeled $4.52$\,dB, $4.57$\,dB, and ${>}5$\,dB would be decoded using the conventional \ac{SC} traversal.

This $\EbNo$-dependent selection suggests a practical way to define \emph{speculative \acp{sp-node}}. Conventional \acp{sp-node} are identified from the frozen bit pattern of a subtree, allowing the decoder to skip standard traversal for known structures. In contrast, speculative \acp{sp-node} would be identified by their ability to pass the speculation verification test with high probability. Thus, a node that does not match a conventional special-node pattern may still be treated as a special case if its speculative decision is consistently reliable.

This provides a possible extension of \ac{sp-node}-based \ac{SC} decoding. In addition to deterministic \acp{sp-node}, the decoder could include speculative \acp{sp-node} selected from the success-rate profile of the decoding tree. These nodes would allow the decoder to exploit reliable speculative regions and further reduce latency while preserving the same final decisions as conventional \ac{SC} decoding.

\begin{figure}[!t]
  \centering
  \begin{tikzpicture}
    \pgfplotsset{
        label style = {font=\fontsize{9pt}{7.2}\selectfont},
        tick label style = {font=\fontsize{7pt}{7.2}\selectfont}
    }
\begin{axis}[
    width=0.9\columnwidth,
    height=0.65\columnwidth,
    xlabel={$N_{\mathrm{lane}}$},
    ylabel={Decoding cycles},
    ylabel style={yshift=-3pt},
    xmode=log,
    log basis x=2,
    xtick={64,128,256,512,1024,2048,4096},
    xticklabels={64,128,256,512,1024,2048,4096},
    grid=major,
    grid style=dashed,
    ytick={500, 750, 1000, 1250, 1500},
    label style={font=\fontsize{9pt}{7.2}\selectfont},
    tick label style={font=\fontsize{7pt}{7.2}\selectfont},
    legend style={
            column sep= 2mm,
            thick,
            font=\fontsize{7pt}{7.2}\selectfont,
        },
]

        \addplot[solid, method/SC, rate/quarter]
        table{
            64             1279
           128             1044
           256              946
           512              909
          1024              897
          2048              894
          4096              894
        };
        \addlegendentry{$\mathcal{P}(4096, 1024)$}

        \addplot[solid, method/SC, rate/half]
        table{
            64             1562
           128             1313
           256             1214
           512             1178
          1024             1166
          2048             1163
          4096             1163
        };
        \addlegendentry{$\mathcal{P}(4096, 2048)$}
        \addplot[solid, method/SC, rate/seveneighths]
        table{
            64              939
           128              707
           256              606
           512              564
          1024              550
          2048              547
          4096              547
        };
        \addlegendentry{$\Pc(4096, 3584)$}

c    \end{axis}
\end{tikzpicture}
  \caption{\ac{SC} decoding cycles across different values of $N_{\mathrm{lane}}$.}
  \label{fig:sc_decode_cycle}
\end{figure}

\begin{figure*}[!htp]
    \centering
    \hspace{25pt}\ref{legend:speedup_per_frame_1024}\\\vspace{5pt}
    \subfigure[$\Pc(4096, 1024)$]{
        \label{fig:speedup_quarter}
        \begin{tikzpicture}
    \pgfplotsset{
        label style = {font=\fontsize{9pt}{7.2}\selectfont},
        tick label style = {font=\fontsize{7pt}{7.2}\selectfont}
    }
    \begin{axis}[
        graph/SpeedUp,
        legend to name=legend:speedup_per_frame_1024_quarter,
        legend columns=5,
        legend pos=north west,
    ]

        \addplot[solid, method/N1024, rate/quarter]
        table{
        1 0.537127
        2 0.671741
        3 1.000884
        4 1.222464
        5 1.350739
        };
        \addlegendentry{$N_v^{\max}=1024$}

        \addplot[solid, method/N512, rate/quarter]
        table{
        1 0.603487
        2 0.845310
        3 1.142317
        4 1.309641
        5 1.437045
        };
        \addlegendentry{$N_v^{\max}=512$}
        
        \addplot[solid, method/N256, rate/quarter]
        table{
        1 0.663870
        2 0.964651
        3 1.200258
        4 1.357219
        5 1.472405
        };
        \addlegendentry{$N_v^{\max}=256$}

        \addplot[solid, method/N128, rate/quarter]
        table{
        1 0.743309
        2 1.046508
        3 1.239068
        4 1.370997
        5 1.450828
        };
        \addlegendentry{$N_v^{\max}=128$}

        \addplot[solid, method/BP, rate/quarter]
        table{
        1      0.727
        1.5    0.997
        2      1.229
        2.5    1.443
        3      1.651
        3.5    1.841
        4      2.039
        4.5    2.200
        5      2.468
        };
        \addlegendentry{BP-20}

    \end{axis}
\end{tikzpicture}
    }
    \subfigure[$\Pc(4096, 2048)$]{
        \label{fig:speedup_half}
        \begin{tikzpicture}
    \pgfplotsset{
        label style = {font=\fontsize{9pt}{7.2}\selectfont},
        tick label style = {font=\fontsize{7pt}{7.2}\selectfont}
    }
    \begin{axis}[
        graph/SpeedUp,
        legend to name=legend:speedup_per_frame_1024_half,
        legend columns=5,
        legend pos=north west,
    ]

        \addplot[solid, method/N1024, rate/half]
        table{
        1 0.566427
        2 0.655211
        3 1.063567
        4 1.585700
        5 1.766699
        };
        \addlegendentry{$N_v^{\max}=1024$}

        \addplot[solid, method/N512, rate/half]
        table{
        1 0.572334
        2 0.754235
        3 1.216375
        4 1.495378
        5 1.624415
        };
        \addlegendentry{$N_v^{\max}=512$}
        
        \addplot[solid, method/N256, rate/half]
        table{
        1 0.594851
        2 0.877809
        3 1.281458
        4 1.495858
        5 1.614202

        };
        \addlegendentry{$N_v^{\max}=256$}

        \addplot[solid, method/N128, rate/half]
        table{
        1 0.635578
        2 0.995242
        3 1.317563
        4 1.480661
        5 1.569776
        };
        \addlegendentry{$N_v^{\max}=128$}

        \addplot[solid, method/BP, rate/half]
        table{
        1    0.684
        1.5  0.788
        2    1.198
        2.5  1.638
        3    2.011
        3.5  2.368
        4    2.690
        4.5  3.049
        5    3.356
        };
        \addlegendentry{BP-20}

    \end{axis}
\end{tikzpicture}
    }
    \subfigure[$\Pc(4096, 3584)$]{
        \label{fig:speedup_7_8}
        \begin{tikzpicture}
    \pgfplotsset{
        label style = {font=\fontsize{9pt}{7.2}\selectfont},
        tick label style = {font=\fontsize{7pt}{7.2}\selectfont}
    }
    \begin{axis}[
        graph/SpeedUp,
        legend to name=legend:speedup_per_frame_1024,
        legend columns=5,
        legend pos=north west,
    ]

        \addplot[solid, method/N1024, rate/seveneighths]
        table{
         3       0.57
         4       0.89
         5       1.57
         6       2.11
         7       2.50
        };
        \addlegendentry{$N_v^{\max}=1024$}

        \addplot[solid, method/N512, rate/seveneighths]
        table{
         3       0.58
         4       0.91
         5       1.54
         6       1.90
         7       2.14
        };
        \addlegendentry{$N_v^{\max}=512$}
        
        \addplot[solid, method/N256, rate/seveneighths]
        table{
         3       0.59
         4       0.90
         5       1.34
         6       1.52
         7       1.62
        };
        \addlegendentry{$N_v^{\max}=256$}

        \addplot[solid, method/N128, rate/seveneighths]
        table{
         3       0.63
         4       0.93
         5       1.23
         6       1.34
         7       1.35
        };
        \addlegendentry{$N_v^{\max}=128$}

        \addplot[solid, method/BP, rate/seveneighths]
        table[col sep=comma]{
        3.000, 0.321
        3.500, 0.385
        4.000, 0.695
        4.500, 1.072
        5.000, 1.415
        5.500, 1.771
        6.000, 2.151
        6.500, 2.602
        7.000, 2.975
        };
        \addlegendentry{BP-20}
        
    \end{axis}
\end{tikzpicture}
    }
    \caption{Decoding speedup of four-core \ac{Spec-SC} and \ac{BP}-20 relative to one-core \ac{SC} for $\Pc(4096, 1024)$, $\Pc(4096, 2048)$ and $\Pc(4096, 3584)$. The four \ac{Spec-SC} cores use $N_{\mathrm{lane}}=(2048,512,256,128)$, with varying maximum speculation node length $N_v^{\max}$. The \ac{BP}-20 decoder uses $N_{\mathrm{lane}}=2048$ to achieve its maximum parallelism. Each point represents the ratio of mean decoding-cycle counts, averaged over 1000 independent frames at the indicated $\EbNo$.}
  \label{fig:speedup}
\end{figure*}

\begin{figure}[!htp]
    \centering
    \hspace{30pt}\includegraphics[width=0.85\linewidth]{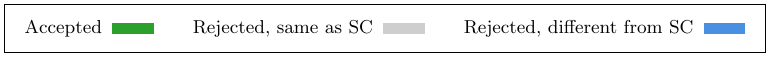}
    \subfigure[$\Pc(4096,1024)$]{
        \includegraphics[width=0.99\linewidth]{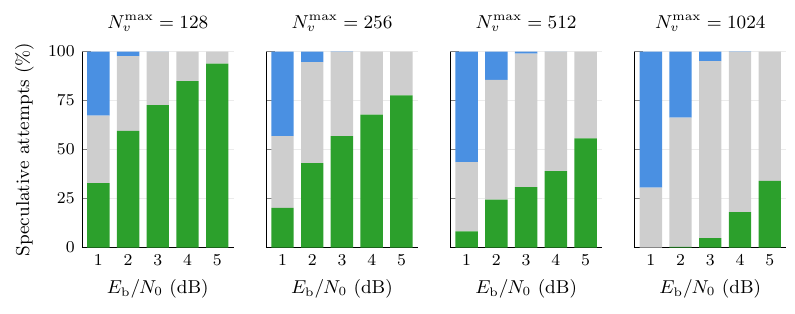}
    }
    \subfigure[$\Pc(4096,2048)$]{
        \includegraphics[width=0.99\linewidth]{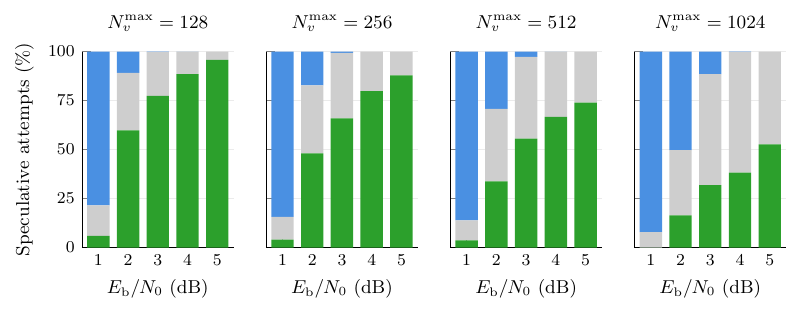}
    }
    \subfigure[$\Pc(4096,3584)$]{
        \includegraphics[width=0.99\linewidth]{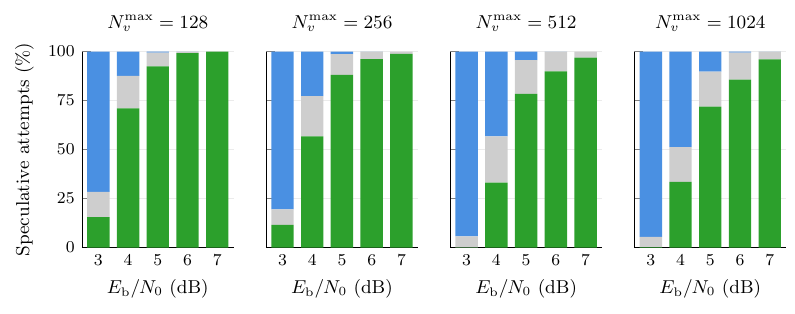}
    }
    \caption{Percentage of speculative attempts classified by verification outcome for $\Pc(4096, 1024)$, $\Pc(4096, 2048)$ and $\Pc(4096, 3584)$ across different $\EbNo$ and maximum speculation node lengths $N_v^{\max}$. Speculative attempts are classified as accepted, rejected despite matching the conventional \ac{SC} decision, or rejected with a different decision. Each category is expressed as a percentage of the total speculative attempts.}
    \label{fig:verify}
\end{figure}

\subsection{Cycle Accounting and Speedup}
\label{sec:multi-core-cycle-accounting}

\begin{figure*}[!t]
    \centering
    \hspace{-15pt}
    \subfigure[$\Pc(4096,1024)$]{
        \label{fig:hist_quarter}
        \includegraphics[width=0.333\linewidth]{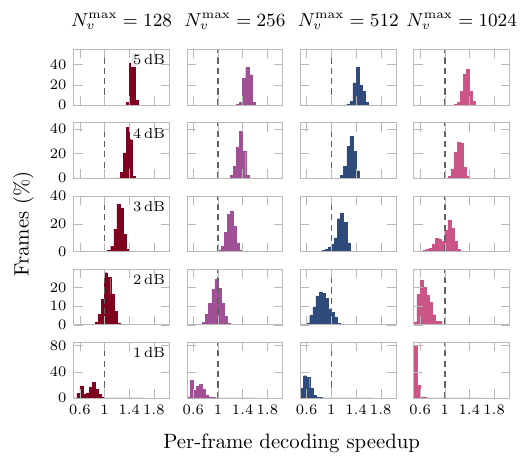}
    }
    \hspace{-17pt}
    \subfigure[$\Pc(4096,2048)$]{
        \label{fig:hist_half}
        \includegraphics[width=0.333\linewidth]{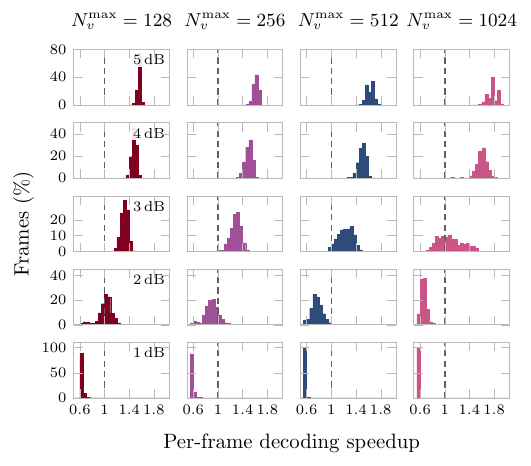}
    }
    \hspace{-17pt}
    \subfigure[$\Pc(4096,3584)$]{
        \label{fig:hist_7_8}
        \includegraphics[width=0.333\linewidth]{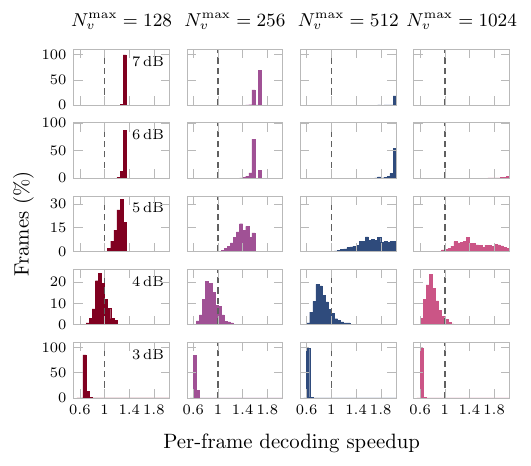}
    }
    \label{fig:frame_speedup}
    \caption{Histograms of per-frame decoding speedup for four-core \ac{Spec-SC} relative to one-core \ac{SC}, evaluated across different code rates, $\EbNo$, and $N_v^{\max}$. }
\end{figure*}

We evaluate selective speculation in a multi-core polar decoder architecture using $C$ cores. Results at each $\EbNo$ are averaged over 1000 independent frames. For each frame, the simulator schedules decoding operations according to their input dependencies.

Each core is modeled with $N_{\mathrm{lane}}$ processing lanes (or processing elements), allowing up to $N_{\mathrm{lane}}$ computations of the same type to be processed in parallel. A decoding operation is represented by a task $\tau$ of type $o_\tau$, and $w_\tau$ is the number of elements that task must process. Its duration, measured in cycles, is
\begin{equation}
  d_\tau =
  \alpha_{o_\tau}
  \left\lceil \frac{w_\tau}{N_{\mathrm{lane}}} \right\rceil,
  \qquad
  \alpha_{o_\tau} =
  \begin{cases}
    2, & o_\tau=\mathrm{validate},\\
    1, & \text{otherwise}.
  \end{cases}
  \label{eq:paper-task-cost}
\end{equation}
The ceiling counts the processing groups, including any partially filled group. Operations other than validation are assigned one cycle per group. For validation, we assume one cycle to synchronize each group of speculative bits and updated \acp{LLR}, followed by one cycle to compare the bits with the hard decisions from those \acp{LLR} and update the acceptance result, giving $\alpha_{\mathrm{validate}}=2$. 

At a recursively decoded node $v$ of length $N_v$, the $f$ and $g$ tasks each compute $N_v/2$ output \acp{LLR}. Combination performs $N_v/2$ bitwise XORs, and validation checks $N_v/2$ bit positions. These tasks therefore use $w_\tau=N_v/2$. Leaf decisions use $w_\tau=1$, while direct fast-node operations and ambiguity-check tasks use $w_\tau=N_v$.

For example, the $f$ task at a node of length $N_v=256$ computes 128 output \acp{LLR}, so $w_\tau=128$. With $N_{\mathrm{lane}}=64$, the core computes 64 outputs in the first cycle and the remaining 64 in the second. Validation also checks 128 positions, but each group of 64 is assigned two cycles, giving four cycles in total. Leaf decisions use $w_\tau=1$. Direct fast-node operations and tie checks are charged for all $N_v$ elements, so they use $w_\tau=N_v$.

\Cref{sec:study_success_rate} shows that speculation is less likely to succeed for larger subnodes. We therefore allow speculation only at recursively decoded parent nodes with $N_v\leq N_v^{\max}$ and $C>1$. This also limits large right-subtree speculations that may take longer than left-subtree decoding.

The right subtree can be decoded speculatively alongside the left subtree, but this speculative decoding does not launch further speculation. The validation task becomes ready after both the speculative decoding and the conditional $g$ update are complete. If validation fails, the relevant \acp{LLR} from the conditional $g$ update are collected, the right subtree is re-decoded, and the updated bits are written back into the decoding result. Each of these correction steps is included in the cycle model.

A task becomes ready when all its required inputs are available. The scheduler assigns ready tasks to available cores, and each task runs on one core until completion. Non-speculative ready tasks take priority over speculative tasks. Within each priority group, the task that became ready first is selected.

Decoding cycles measure the time until the estimated codeword is available. Let $\mathcal{T}_{\mathrm{dec}}$ contain all executed decoding tasks, including rejected speculative work, validation, re-decoding, input collection, writing back corrected bits, fast-node checks, and fallbacks. If $e_\tau$ is the completion time of task $\tau$ measured from the start of decoding, then
\begin{equation}
  T_{\mathrm{dec}}=\max_{\tau\in \mathcal{T}_{\mathrm{dec}}}e_\tau.
  \label{eq:paper-latency}
\end{equation}
This includes time spent waiting for inputs or a free core and excludes subsequent output processing.

For the \ac{BP} decoder, one iteration consists of a right-to-left and a left-to-right traversal of the $\log_2 N$ factor-graph stages. Each stage updates $N$ \ac{LLR} messages. Under the same $N_{\mathrm{lane}}$ processing-width model, the cycle count for a frame requiring $I$ iterations is therefore modeled as
\begin{equation}
T_{\mathrm{BP}}(I) = 2I\sum_{i=0}^{\log_2 N-1} \left\lceil \frac{N_i}{N_{\mathrm{lane}}} \right\rceil,
\end{equation}
where $N_i$ is the number of nontrivial processing elements in stage $i$. Let $\overline{T}_{\mathrm{SC}}$ and $\overline{T}_{\mathrm{dec}}(C,N_v^{\max})$ denote the mean decoding-cycle counts of one-core \ac{SC} and another arbitrary decoder, respectively.

Both decoders use the same received frames and fast-node settings, with speculation disabled for the \ac{SC} baseline. The one-core \ac{SC} baseline uses the same processing width, $N_{\mathrm{lane}}$, as Core 1 of the four-core \ac{Spec-SC} decoder. The speedup is
\begin{equation}
  \operatorname{Speedup}(C,N_v^{\max}) =
  \frac{\overline{T}_{\mathrm{SC}}}{\overline{T}_{\mathrm{dec}}(C,N_v^{\max})}.
  \label{eq:paper-sc-speedup}
\end{equation}
A value above one indicates faster decoding than \ac{SC}. The decoding cycles required for \ac{SC} for different values of $N_{\mathrm{lane}}$ are shown in \Cref{fig:sc_decode_cycle}.

To determine whether speculation can further reduce latency when increasing the processing width no longer benefits conventional \ac{SC}, we compare four-core \ac{Spec-SC} with one-core \ac{SC} across different code rates, as shown in~\Cref{fig:speedup}. The four \ac{Spec-SC} cores use heterogeneous processing widths, with $N_{\mathrm{lane}}=2048$, $512$, $256$, and $128$ lanes for Cores 1--4, respectively. For the one-core \ac{SC} baseline with $N_{\mathrm{lane}}=2048$, each decoding task fits within a single processing group in our model. Therefore, further increasing $N_{\mathrm{lane}}$ does not reduce the number of decoding cycles for \ac{SC}.

For $\Pc(4096,1024)$, \ac{Spec-SC} achieves a speedup of approximately $1.45\times$ at $5\,\mathrm{dB}$ with $N_v^{\max}=128$, compared with $1.35\times$ for $N_v^{\max}=1024$. In contrast, for $\Pc(4096,2048)$, larger speculation limits become more beneficial as $\EbNo$ increases. At $5\,\mathrm{dB}$, the speedup reaches approximately $1.75\times$ for $N_v^{\max}=1024$, compared with $1.55\times$ for $N_v^{\max}=128$. These results indicate that the preferred speculation node size depends on both the code rate and operating $\EbNo$.

For $\Pc(4096,3584)$, \ac{Spec-SC} achieves latency comparable to \ac{SC} at $4\,\mathrm{dB}$. Beyond this point, \ac{Spec-SC} provides increasing gains. At $5\,\mathrm{dB}$, where the \ac{FER} is on the order of $10^{-4}$, \ac{Spec-SC} achieves a $1.57\times$ speedup for $N_v^{\max}=1024$. The speedup increases to $2.5\times$ at $7\,\mathrm{dB}$. Although the corresponding \ac{FER} is below the operating region considered here, this result demonstrates the increasing effectiveness of speculation at higher $\EbNo$. The four-core configuration uses a total of $2944$ processing lanes, representing a $43.75\%$ increase over the $N_{\mathrm{lane}}=2048$ \ac{SC} baseline. These results demonstrate that additional speculative cores can further reduce decoding latency even when conventional \ac{SC} has fully exploited lane parallelism.

The \ac{BP}-20 decoder achieves speedups of approximately $2.5\times$ and $3.3\times$ at $5\,\mathrm{dB}$ for $\Pc(4096,1024)$ and $\Pc(4096,2048)$, respectively, exceeding those of \ac{Spec-SC}. For $\Pc(4096,3584)$, \ac{BP}-20 achieves a lower speedup than \ac{Spec-SC} at $5\,\mathrm{dB}$ but reaches approximately $3.0\times$ at $7\,\mathrm{dB}$. However, as shown in~\Cref{fig:4k_ber_fer}, the error-correction performance of \ac{BP} diverges from that of \ac{SC} and \ac{Spec-SC} at higher $\EbNo$. Thus, the higher \ac{BP} speedups in this region are not achieved at equivalent error-correction performance.

To better understand the factors limiting the speedup observed in~\Cref{fig:speedup}, \Cref{fig:verify} examines the outcomes of speculative attempts. Each attempt is classified as accepted, rejected despite matching the conventional \ac{SC} decision, or rejected with a different decision. For $\Pc(4096,3584)$ with $N_v^{\max}=1024$, the acceptance rate increases from $33.6\%$ at $4\,\mathrm{dB}$ to $71.9\%$ at $5\,\mathrm{dB}$ and $96.1\%$ at $7\,\mathrm{dB}$. Over the same range, the proportion of rejected decisions that differ from \ac{SC} decreases from $48.7\%$ to $10.2\%$ and approximately zero, respectively. This corresponds to the increasing decoding speedup observed in~\Cref{fig:speedup}, as fewer speculative branches require re-decoding.

However, a considerable proportion of rejected speculative decisions still match the conventional \ac{SC} output. At $5\,\mathrm{dB}$ with $N_v^{\max}=1024$, this accounts for $66.0\%$ and $47.4\%$ of speculative attempts for $\Pc(4096,1024)$ and $\Pc(4096,2048)$. These results show that the verification criterion can reject otherwise correct speculative decisions, requiring unnecessary re-decoding and limiting the achievable speedup.

To examine how the speculative outcomes translate into decoding latency, \Cref{fig:frame_speedup} shows the distribution of per-frame decoding speedup. This distribution shows how frequently speculation reduces or increases latency for individual frames. For $\Pc(4096,3584)$ with $N_v^{\max}=1024$, only $4.1\%$ of frames achieve a speedup at $4\,\mathrm{dB}$. This proportion increases to $98.6\%$ at $5\,\mathrm{dB}$. A similar trend is observed for $\Pc(4096,1024)$ and $\Pc(4096,2048)$, where $64.2\%$ and $64.0\%$ of frames, respectively, achieve a speedup at $3\,\mathrm{dB}$, increasing to $100\%$ at $5\,\mathrm{dB}$. These results show that \ac{Spec-SC} can achieve higher throughput as $\EbNo$ increases.

The node-size limit provides a simple way to control speculation. However, verification success also depends on the node and operating $\EbNo$, as shown in \Cref{sec:study_success_rate}. Selecting nodes using these success rates could provide further speedup. We leave this approach for future work, as it requires further study of node selection and scheduling for hardware implementation.

\section{Conclusion}
\label{sec:conclusion}
This paper presents \ac{Spec-SC}, a speculative decoding framework that extends conventional \ac{SC} decoding through a verification-based speedup mechanism. In contrast to conventional \ac{sp-node}-based decoding, where parallelism is limited to collapsing predetermined special-node structures, \ac{Spec-SC} enables additional parallelism in general nodes that do not necessarily match existing special-node patterns. By allowing the right branch of a decoding subtree to be processed speculatively and subsequently verified, the proposed framework reduces the effective traversal of otherwise sequential regions of the \ac{SC} decoding tree. 

Due to the verification, \ac{Spec-SC} preserves the error-correction performance of the conventional \ac{SC} decoder. Simulation results further demonstrate that the average number of nodes traversed in the main decoding tree is substantially reduced. For a polar code of length $N=4096$, node-traversal reductions of $64\%$, $69\%$ and $66\%$ are observed for code rates $R=\nicefrac{1}{4}$, $R=\nicefrac{1}{2}$, and $R=\nicefrac{7}{8}$, respectively, at an FER target of $10^{-2}$. 

The ablation study further shows that speculation is not equally effective at all nodes. Instead, the verification success rate depends on both the node location in the decoding tree and the channel condition. As $\EbNo$ increases, some nodes consistently achieve high verification success rates, indicating that their speculative decisions are often identical to the decisions produced by conventional \ac{SC} decoding. This suggests that speculation should not necessarily be applied uniformly across the entire tree. Rather, the success-rate profile can be used to identify reliable speculative nodes.

The cycle-level analysis shows that reduced sequential traversal can translate into lower modeled decoding latency when additional cores are available. For $\Pc(4096,3584)$, the heterogeneous four-core \ac{Spec-SC} configuration achieves a $1.57\times$ speedup at $5\,\mathrm{dB}$, where the \ac{FER} is on the order of $10^{-4}$, relative to a one-core \ac{SC} baseline. This improvement is achieved with $43.75\%$ more processing lanes. The speedup increases to $2.5\times$ at $7\,\mathrm{dB}$, demonstrating further latency reduction at higher $\EbNo$. These results demonstrate a latency benefit even after further increasing $N_{\mathrm{lane}}$ no longer reduces SC decoding cycles. However, speculation can increase latency at lower $\EbNo$. Its benefit therefore depends on $N_{\mathrm{lane}}$, the speculation node size, and the operating $\EbNo$, motivating selective speculation.

Overall, the \ac{Spec-SC} algorithm is technology independent and could be implemented in various computing platforms such as \acp{CPU}, \acp{GPU}, \acp{FPGA} and \acp{ASIC}. 

\section{Future Work}
These findings motivate the concept of \textit{speculative special-nodes}. Unlike conventional \acp{sp-node}, which are defined by fixed frozen bit patterns, speculative \acp{sp-node} would be identified based on their likelihood of passing the speculation verification test, rather than on structural criteria alone. Such nodes would allow future decoders to skip additional tree traversal even for subtrees that do not conform to a conventional special-node structure, thereby extending the applicability of speculation beyond the cases considered in this work.

A second direction concerns the verification criterion itself. Although the proposed criterion guarantees no loss in error-rate performance, our results show that it frequently fails to recognize speculative results that are in fact correct, limiting the achievable speedup. Developing more optimistic (yet still reliable) verification criteria is therefore a promising avenue for improving performance further.

This work has also been restricted to the algorithmic level of \ac{Spec-SC}. An important next step is to evaluate the scheme in an actual hardware implementation to assess how speculation affects throughput and hardware complexity in practice, since these effects cannot be fully captured through algorithmic analysis alone.

Finally, a detailed timing analysis of \ac{Spec-SC} is needed, along with an investigation of its scheduling in fully-parallel and semi-parallel decoder architectures, both of which are likely to reveal additional trade-offs relevant to practical deployment.

\bibliographystyle{IEEEtran}
\bibliography{IEEEabrv,bibliography}

\begin{thebibliography}{10}
\providecommand{\url}[1]{#1}
\csname url@samestyle\endcsname
\providecommand{\newblock}{\relax}
\providecommand{\bibinfo}[2]{#2}
\providecommand{\BIBentrySTDinterwordspacing}{\spaceskip=0pt\relax}
\providecommand{\BIBentryALTinterwordstretchfactor}{4}
\providecommand{\BIBentryALTinterwordspacing}{\spaceskip=\fontdimen2\font plus
\BIBentryALTinterwordstretchfactor\fontdimen3\font minus \fontdimen4\font\relax}
\providecommand{\BIBforeignlanguage}[2]{{%
\expandafter\ifx\csname l@#1\endcsname\relax
\typeout{** WARNING: IEEEtran.bst: No hyphenation pattern has been}%
\typeout{** loaded for the language `#1'. Using the pattern for}%
\typeout{** the default language instead.}%
\else
\language=\csname l@#1\endcsname
\fi
#2}}
\providecommand{\BIBdecl}{\relax}
\BIBdecl

\bibitem{arikan2009channel}
E.~Ar{\i}kan, ``{Channel polarization: A method for constructing capacity-achieving codes for symmetric binary-input memoryless channels},'' \emph{{IEEE} Trans. Inf. Theory}, vol.~55, no.~7, pp. 3051--3073, Jul. 2009.

\bibitem{sarkis2014fast}
G.~Sarkis, P.~Giard, A.~Vardy \emph{et~al.}, ``{Fast polar decoders: Algorithm and implementation},'' \emph{{IEEE} J. Sel. Areas Commun.}, vol.~32, no.~5, pp. 946--957, 2014.

\bibitem{sarkis2015fast}
------, ``Fast list decoders for polar codes,'' \emph{{IEEE} J. Sel. Areas Commun.}, vol.~34, no.~2, pp. 318--328, Nov. 2015.

\bibitem{tal2015list}
I.~Tal and A.~Vardy, ``List decoding of polar codes,'' \emph{{IEEE} Trans. Inf. Theory}, vol.~61, no.~5, pp. 2213--2226, May 2015.

\bibitem{fan2015low}
Y.~Fan, C.~Xia, J.~Chen \emph{et~al.}, ``A low-latency list successive-cancellation decoding implementation for polar codes,'' \emph{{IEEE} J. Sel. Areas Commun.}, vol.~34, no.~2, pp. 303--317, 2015.

\bibitem{hashemi2017fast}
S.~A. Hashemi, C.~Condo, and W.~J. Gross, ``Fast and flexible successive-cancellation list decoders for polar codes,'' \emph{{IEEE} Trans. Signal Process.}, vol.~65, no.~21, pp. 5756--5769, Aug. 2017.

\bibitem{alamdar2011simplified}
A.~Alamdar-Yazdi and F.~R. Kschischang, ``A simplified successive-cancellation decoder for polar codes,'' \emph{{IEEE} Commun. Lett.}, vol.~15, no.~12, pp. 1378--1380, 2011.

\bibitem{balatsoukas2015llr}
A.~Balatsoukas-Stimming, M.~B. Parizi, and A.~Burg, ``{LLR-based successive cancellation list decoding of polar codes},'' \emph{{IEEE} Trans. Signal Process.}, vol.~63, no.~19, pp. 5165--5179, 2015.

\bibitem{cocskun2022information}
M.~C. Co{\c{s}}kun and H.~D. Pf{\i}ster, ``An information-theoretic perspective on successive cancellation list decoding and polar code design,'' \emph{{IEEE} Trans. Inf. Theory}, vol.~68, no.~9, pp. 5779--5791, 2022.

\bibitem{chen2012list}
K.~Chen, K.~Niu, and J.~Lin, ``List successive cancellation decoding of polar codes,'' \emph{Electron. Lett.}, vol.~48, no.~9, pp. 500--501, 2012.

\bibitem{Niu2012CRC}
K.~Niu and K.~Chen, ``{CRC-aided decoding of polar codes},'' \emph{{IEEE} Commun. Lett.}, vol.~16, no.~10, pp. 1668--1671, Sep. 2012.

\bibitem{li2012adaptive}
B.~Li, H.~Shen, and D.~Tse, ``An adaptive successive cancellation list decoder for polar codes with cyclic redundancy check,'' \emph{{IEEE} Commun. Lett.}, vol.~16, no.~12, pp. 2044--2047, 2012.

\bibitem{sarkis2012data}
G.~Sarkis and W.~J. Gross, ``Polar codes for data storage applications,'' in \emph{International Conference on Computing, Networking and Commun. (ICNC 2013)}, pp. 840--844.

\bibitem{zhang2015split}
Z.~Zhang, L.~Zhang, X.~Wang \emph{et~al.}, ``A split-reduced successive cancellation list decoder for polar codes,'' \emph{{IEEE} J. Sel. Areas Commun.}, vol.~34, no.~2, pp. 292--302, 2015.

\bibitem{zhang2023channel}
H.~Zhang and W.~Tong, ``Channel coding for {6G} extreme connectivity--requirements, capabilities, and fundamental tradeoffs,'' \emph{IEEE BITS the Information Theory Magazine}, vol.~3, no.~1, pp. 54--66, 2023.

\bibitem{geiselhart20236g}
M.~Geiselhart, F.~Krieg, J.~Clausius \emph{et~al.}, ``{6G}: A welcome chance to unify channel coding?'' \emph{IEEE BITS the Information Theory Magazine}, vol.~3, no.~1, pp. 67--80, 2023.

\bibitem{proceedings2024trends}
S.~Miao, C.~Kestel, L.~Johannsen \emph{et~al.}, ``{Trends in Channel Coding for 6G},'' \emph{Proceedings of the IEEE}, vol. 112, no.~7, pp. 653--675, 2024.

\bibitem{yang2026improved}
Z.~Yang, L.~Chen, X.~Wang \emph{et~al.}, ``{Improved successive cancellation decoding of long polar codes through perturbing a posteriori LLRs and its theoretical insights},'' \emph{{IEEE} Trans. Commun.}, 2026.

\bibitem{wang2023perturbation}
X.~Wang, H.~Zhang, J.~Tong \emph{et~al.}, ``{Perturbation-enhanced SCL decoder for polar codes},'' in \emph{2023 IEEE Globecom Workshops (GC Wkshps)}, 2023, pp. 1674--1679.

\bibitem{li2025enhanced}
J.~Li, S.~Shen, and W.~J. Gross, ``Enhanced successive cancellation list decoder for long polar codes targeting {6G} air interface,'' \emph{arXiv preprint arXiv:2508.16498}, 2025.

\bibitem{mehmood2021experimental}
T.~Mehmood, M.~P. Yankov, H.~E. Hansen \emph{et~al.}, ``Experimental demonstration of rate-adaptive concatenated codes for flexible optical networks,'' \emph{{IEEE} Photon. Technol. Lett.}, vol.~33, no.~24, pp. 1447--1450, 2021.

\bibitem{koike2018irregular}
T.~Koike-Akino, C.~Cao, Y.~Wang \emph{et~al.}, ``Irregular polar coding for complexity-constrained lightwave systems,'' \emph{J. Lightw. Technol.}, vol.~36, no.~11, pp. 2248--2258, 2018.

\bibitem{krieg2025longpolarvsldpc}
F.~Krieg, M.~R{\"u}benacke, A.~Zunker \emph{et~al.}, ``{Long Polar vs. LDPC Codes under Complexity-Constrained Decoding},'' in \emph{2025 IEEE Globecom Workshops (GC Wkshps)}, 2025, pp. 272--277.

\bibitem{tong2023fast}
J.~Tong, X.~Wang, Q.~Zhang \emph{et~al.}, ``{Fast polar codes for terabits-per-second throughput communications},'' in \emph{2023 IEEE 34th Annual International Symposium on Personal, Indoor and Mobile Radio Communications (PIMRC)}, 2023, pp. 1--6.

\bibitem{ren2022sequence}
Y.~Ren, A.~T. Kristensen, Y.~Shen \emph{et~al.}, ``{A sequence repetition node-based successive cancellation list decoder for 5G polar codes: Algorithm and implementation},'' \emph{{IEEE} Trans. Signal Process.}, vol.~70, pp. 5592--5607, 2022.

\bibitem{zheng2021threshold}
H.~Zheng, S.~A. Hashemi, A.~Balatsoukas-Stimming \emph{et~al.}, ``Threshold-based fast successive-cancellation decoding of polar codes,'' \emph{{IEEE} Trans. Commun.}, vol.~69, no.~6, pp. 3541--3555, 2021.

\bibitem{arikan2008performance}
E.~Arikan, ``{A performance comparison of polar codes and Reed-Muller codes},'' \emph{{IEEE} Commun. Lett.}, vol.~12, no.~6, pp. 447--449, 2008.

\bibitem{niu2014polar}
K.~Niu, K.~Chen, J.~Lin \emph{et~al.}, ``{Polar codes: Primary concepts and practical decoding algorithms},'' \emph{{IEEE} Commun. Mag.}, vol.~52, no.~7, pp. 192--203, 2014.

\bibitem{yuan2014early}
B.~Yuan and K.~K. Parhi, ``Early stopping criteria for energy-efficient low-latency belief-propagation polar code decoders,'' \emph{{IEEE} Trans. Signal Process.}, vol.~62, no.~24, pp. 6496--6506, 2014.

\bibitem{abbas2016high}
S.~M. Abbas, Y.~Fan, J.~Chen \emph{et~al.}, ``{High-throughput and energy-efficient belief propagation polar code decoder},'' \emph{{IEEE} Trans. {VLSI} Syst.}, vol.~25, no.~3, pp. 1098--1111, 2016.

\bibitem{sun2016error}
S.~Sun, S.-G. Cho, and Z.~Zhang, ``Error patterns in belief propagation decoding of polar codes and their mitigation methods,'' in \emph{50th Asilomar Conference on Signals, Systems and Computers}.\hskip 1em plus 0.5em minus 0.4em\relax IEEE, 2016, pp. 1199--1203.

\bibitem{elkelesh2017errorfloor}
A.~Elkelesh, S.~Cammerer, M.~Ebada \emph{et~al.}, ``Mitigating clipping effects on error floors under belief propagation decoding of polar codes,'' in \emph{2017 International Symposium on Wireless Communication Systems (ISWCS)}, 2017, pp. 384--389.

\bibitem{bossert2025hiddencodewords}
M.~Bossert, ``{Soft Decision Decoding of Recursive Plotkin Constructions Based on Hidden Code Words},'' \emph{{IEEE} Trans. Inf. Theory}, vol.~71, no.~6, pp. 4228--4249, 2025.

\bibitem{kamenev2023highratefirst}
M.~Kamenev, ``{Recursive Decoding of Reed-Muller Codes Starting With the Higher-Rate Constituent Code},'' \emph{{IEEE} Trans. Inf. Theory}, vol.~69, no.~4, pp. 2206--2217, 2023.

\bibitem{kam2021ultralow}
D.~Kam, H.~Yoo, and Y.~Lee, ``Ultralow-latency successive cancellation polar decoding architecture using tree-level parallelism,'' \emph{{IEEE} Trans. {VLSI} Syst.}, vol.~29, no.~6, pp. 1083--1094, 2021.

\bibitem{yoo2016efficient}
H.~Yoo and I.-C. Park, ``Efficient pruning for successive-cancellation decoding of polar codes,'' \emph{{IEEE} Commun. Lett.}, vol.~20, no.~12, pp. 2362--2365, 2016.

\bibitem{mishra2012successive}
A.~Mishra, A.~J. Raymond, L.~G. Amaru \emph{et~al.}, ``{A successive cancellation decoder ASIC for a 1024-bit polar code in 180nm CMOS},'' in \emph{2012 IEEE Asian solid state circuits conference (A-SSCC)}, pp. 205--208.

\bibitem{zhou2016segmented}
H.~Zhou, C.~Zhang, W.~Song \emph{et~al.}, ``{Segmented CRC-aided SC list polar decoding},'' in \emph{83rd Vehicular Technology Conference (VTC Spring)}.\hskip 1em plus 0.5em minus 0.4em\relax IEEE, 2016, pp. 1--5.

\bibitem{hashemi2016partion}
S.~A. Hashemi, A.~Balatsoukas-Stimming, P.~Giard \emph{et~al.}, ``Partitioned successive-cancellation list decoding of polar codes,'' in \emph{IEEE International Conference on Acoustics, Speech and Signal Processing (ICASSP)}, 2016, pp. 957--960.

\bibitem{hashemi2017partition}
S.~A. Hashemi, C.~Condo, F.~Ercan \emph{et~al.}, ``Memory-efficient polar decoders,'' \emph{IEEE J. Emerg. Sel. Top. Circuits Syst.}, vol.~7, no.~4, pp. 604--615, 2018.

\bibitem{sionna}
J.~Hoydis, S.~Cammerer, F.~{Ait Aoudia} \emph{et~al.}, ``Sionna: An open-source library for next-generation physical layer research,'' \emph{arXiv preprint}, 2022.

\bibitem{Trifonov2012Efficient}
P.~Trifonov, ``Efficient design and decoding of polar codes,'' \emph{{IEEE} Trans. Commun.}, vol.~60, no.~11, pp. 3221--3227, Aug. 2012.

\end{thebibliography}
\end{document}